\documentclass[leqno]{amsart}

\usepackage[english]{babel} 
\usepackage{hyperref} 
\usepackage{srcltx}
\usepackage{color}
\usepackage{amsmath}
\usepackage{graphicx}

\newtheorem{defin}{Definition}[section]
\newtheorem{theorem}{Theorem}[section]

\newtheorem{remark}{Remark}[section]

\newcommand{\N}{\mathbb{N}}

\newcommand{\X}{\times}

\newcommand\D{\partial}

\newcommand{\spann}{\mathop\mathrm{span}\nolimits}

\AtBeginDocument{
  \let\div\undefined\DeclareMathOperator{\div}{div}
}
\AtBeginDocument{
  \let\curl\undefined\DeclareMathOperator{\curl}{curl}
}

\numberwithin{equation}{section}

\begin{document}
\title[Quasicrystal dynamics with non-linear gyroscopic effects]{Existence 
results in the linear dynamics of quasicrystals with phason diffusion and non-linear gyroscopic effects}
\author{Luca Bisconti*, Paolo Maria Mariano**}
\address{*DiMaI ``U. Dini'', Universit\`a di Firenze, V.le Morgagni 67/A, I-50134 Firenze}
\email{luca.bisconti@unifi.it}
\address{**DiCeA,  Universit\`a di Firenze, V. Santa Marta 3, I-50139 Firenze}
\email{paolo.mariano@unifi.it}
\date{\today}

\begin{abstract}
Quasicrystals are characterized by quasi-periodic arrangements of atoms. The description of their mechanics involves deformation and a (so-called \emph{phason}) vector field accounting at macroscopic scale of local phase changes, due to atomic flips necessary to match quasi-periodicity under the action of the external environment. Here we discuss the mechanics of quasicrystals, commenting the shift from its initial formulation, as standard elasticity in a space with dimension twice the ambient one, to a more elaborated setting avoiding physical inconveniences of the original proposal. In the new setting we tackle two problems. First we discuss the linear dynamics of quasicrystals including a phason diffusion. We prove existence of weak solutions and their uniqueness under rather general boundary and initial conditions. We then consider phason rotational inertia, non-linearly coupled with the $\mathrm{curl}$ of the macroscopic velocity, and prove once again existence of weak solutions to the pertinent balance equations.
\end{abstract}

\maketitle

\section{Introduction} \label{sec:introduction}
\subsection{Quasi-periodic atomic arrangements}

In a 1991 report of the International Union of Crystallography, under
\textquotedblleft terms of reference\textquotedblright , we
find that \textquotedblleft by `crystal' we mean any solid having an
essentially discrete diffraction diagram, and by `aperiodic crystal' we mean
any crystal in which three-dimensional lattice periodicity can be considered
to be absent\textquotedblright\ \cite{ICU}. The definition includes then the
possibility of quasi-periodic crystals, also called \emph{quasicrystals} to
remind the circumstance.

Such a viewpoint overcame in the definition of crystals the previous mention of periodicity of the 
atomic distribution in space: a sliding in the
basic paradigm of crystallography induced by the 1982 experimental discovery
by D. Shechtman, published in 1984 \cite{She}, of the possibility of atomic
arrangements with icosahedral symmetry in aluminium-based synthetic alloys,
and penthagonal symmetry in thin films of the same materials, not determined by twinned atomic structures. Natural
quasi-periodic alloys have been found since then in meteorites.

A typical example of a quasiperiodic function over the line is $\sin x+\cos
\alpha x$, with $\alpha $ an irrational number. It can be intended as
obtained from a periodic function in the plane, namely $\sin x+\cos y$, with
the additional constraint $y=\alpha x$. Moreover, if we consider a
quasi-periodic distribution of mass points in $3D$ space and compute the
Fourier expansion of the mass distribution, we find that the emerging wave
vector is \emph{six}-dimensional. In general, a quasi-periodic atomic
arrangement can be viewed as the orthogonal projection of a portion of a
periodic lattice onto an appropriate (incommensurate) subspace. The
recurrent example includes a periodic lattice with square symmetry, filling a plane. If
we select a strip around a straight line, inclined about $3/2$ with respect
to the symmetry axes, and we project orthogonally the atoms in the strip
over the line, we get a periodic one-dimensional lattice. In contrast, if we
choose a straight line inclined by an irrational angle and reproduce the
same process, what we find is a quasi-periodic one-dimensional lattice.
Quasi-periodicity emerges even in the golden mean case (i.e. $3/2$) if we
move orthogonally to the line the lattice in an appropriate way: some atoms
go out the strip, others enter it. Such orthogonal shifts take the name of 
\textbf{phason} defects. The common terminology seems to recall that, by
means of the orthogonal displacement of the lattice, we are changing the
phase (at least with respect to the symmetry) of the lattice over the line.
Such a construction, however, is just an ideal geometric representation of
what is in nature. In the physical space, we can consider phasons as \emph{inner} (low
spatial scale) degrees of freedom exploited locally by the atoms to assure
quasi-periodicity, compatibly with the boundary conditions imposed to a
quasi-crystalline body by the interaction with the external environment. As
inner degrees of freedom, they are invariant with respect to rigid
translations of the whole body. This aspect has a key role in the
developments presented in what follows.

\subsection{Origins of and trends in the mechanics of quasicrystals}

In building up continuum models of the mechanics of quasicrystals, the
geometric constructions above have suggested at a first glance of proposing
just a higher dimensional replica of crystal elasticity traditional setting (see, e.g., for
the first proposals \cite{Lubensky}, \cite{DP} and also \cite{Jeong}, \cite{RL}, \cite{HWD}). Researchers have then focused the attention
primarily on linear elasticity, meeting this way the concrete advantage of
having at disposal a format where we can easily reproduce sistematically,
without thecnical and conceptual difficulties, all the standard results of
traditional linear elasticity. The approach may lead to conclusions
with peculiar physical significance (a review of the production in this trend
is in \cite{Fan}; up-dated results are in \cite{Li}, \cite{Li-Yun14}, \cite{LWWC14}, \cite{Lian-He13}, \cite{RM11b}). However, in 2011, S.
Colli and P. M. Mariano \cite{CM11} showed the existence of at least two
cases where such a format of quasicrystal linear elasticity produces
non-physical results, i.e. the instantaneous propagation at infinity of the
phason disturbances, described at a continuum level by a differentiable
vector field (this way the representation is multi-scale) and they
conjectured that a non-zero conservative phason self-action would avoid such a drawback,
assuring phason decay in space. The analytical proof of the conjecture for
one-dimensional quasicrystals is in \cite{MP13}. Numerical
simulations corroborating the conjecture in two-dimensional space appeared
later in \cite{MS15}.

\begin{itemize}
\item In fact, a dissipative phason self-action has been already \emph{assumed} to
exist in \cite{RL} but just with dissipative nature, an action presumed to
drive phason diffusion. The problems evidenced in \cite{CM11} were, however,
in purely conservative static setting.

\item A first proof of the existence of a possibly nonzero phason
self-action with both conservative and dissipative nature appeared
first in \cite{M06}. Then it was rediscussed in \cite{MP13}, where an ancillary
consequence shows the theoretical possibility of a
rotational-type phason inertia induced by the local spin of the macroscopic
velocity field, an aspect taken into account here.

\end{itemize}

\subsection{What we discuss}

Here we focus the attention on the dynamics of quasicrystals in small-strain
regime, including a phason self-action with both conservative and dissipative components, and gyroscopic type phason inertia effects. Before tackling the analysis of the pertinent balance equations, in \textbf{Section 2} we rediscuss preliminarly their
deduction from
the invariance under rigid-body changes in observers of the power of
external actions over a generic part of the body. We briefly reproduce the
path followed in \cite{MP13} by showing in addition a non-standard
action-reaction principle for the phason traction and the existence of the
phason stress (Cauchy-type theorem), which emerge as special occurrences of
abstract results in the general model-building framework of the mechanics of
complex materials, presented in \cite{M14}. In contrast with what is
developed in \cite{MP13}, we restrict the treatment just to Euclidean
frames, identifying covariant and contravariant components of the tensors
considered, because the special cases tackled analytically in the subsequent
sections refer to that frames.

In \textbf{Section 3},  we consider in small strain regime (1)
linear constitutive structures for the Cauchy stress, the phason stress, and
the conservative component of the phason self-action, (2) phason diffusion driven
by a dissipative phason self-action, (3) macroscopic inertia. We provide existence and uniqueness theorems for the weak solution of the balance
equations. Then, in \textbf{Section 4}, we consider non-linear rotational-type phason
inertia and we provide a theorem of existence (and regularity) of the
weak solutions of these modified balance equations. In the treatment we consider
first regularization induced by viscous-type standard and phason stresses.
Then we compute the limit when such regularizations vanish. The results in
the linear case are a necessary prerequisite for the non-linear one.

\section{Continuum mechanics of quasicrystals}

\subsection{Deformation and phason field}

We write $\mathcal{B}$ for the macroscopic \textbf{reference} shape  of a quasicrystalline body (it is just a
geometric setting where we may compare lenghts, angles, surfaces and volumes
to measure strain), assumed to be a bounded
arcwise connected open region in the three-dimensional point space $\mathcal{E%
}^{3}$, coinciding with the interior of its closure and endowed with
surface-like boundary uniquely oriented everywhere to within a finite number of corners
and edges. In another space, indicated by $\mathcal{\tilde{E}}^{3}$ and
distinguished by $\mathcal{E}^{3}$ just by an isomorphism $i:\mathcal{E}%
^{3}\longrightarrow \mathcal{\tilde{E}}^{3}$, which we can choose as an
orientation preserving isometry or even the simple identification, we record
shapes of the body that we consider \textbf{deformed} with respect to $%
\mathcal{B}$, reached by means of one-to-one, differentiable, orientation
preserving maps $x\longmapsto y:=\tilde{y}\left( x\right) \in \mathcal{%
\tilde{E}}^{3}$.

The distinction between the two spaces justifies the standard statement
that two observers, i.e. two frames in the \emph{whole} space, differing one
another by a rigid-body motion, evaluate the \emph{same} reference place.
Moreover, such a distinction is crucial when we want to consider material
mutations, which are naturally described by a non-unique choice of the
reference place (see \cite{M14}).

A field taking values in a three-dimensional real vector space $\mathcal{V}%
^{3}$, precisely $x\longmapsto \nu :=\tilde{\nu}\left( x\right) \in \mathcal{%
V}^{3}$, assumed to be differentiable, accounts point-by-point at the
continuum scale for the atomic flips, which allow to match
quasi-periodicity. This is the so-called \textbf{phason field} in Lagrangian
representation, i.e. considered as a field over the reference place. We then
call $\mathcal{V}^{3}$ the \textbf{phason space}.

From now on we endow $\tilde{\mathcal{E}}^{3}$, $\mathcal{E}^{3}$ and 
$\mathcal{V}^{3}$ with Cartesian frames.

\textbf{Motions} are then (in generalized sense) pairs%
\begin{equation*}
\left( x,t\right) \longmapsto y:=\tilde{y}\left( x,t\right) \in \mathcal{%
\tilde{E}}^{3},\text{ \ \ }\left( x,t\right) \longmapsto \nu :=\tilde{\nu}%
\left( x,t\right) \in \mathcal{V}^{3},
\end{equation*}%
assumed to be sufficiently differentiable in time.

We shall write $F$ and $N$ for the \textbf{deformation gradient} and the 
\textbf{phason field gradient}, evaluated at $x$ and $t$. The assumption
that the deformation preserves the local orientation of triples of linearly
independent vectors implies $\det F>0$, a standard consequence, indeed. We
define another vector field, the \textbf{displacement}, as%
\begin{equation*}
\left( x,t\right) \longmapsto u:=\tilde{u}\left( x,t\right) :=\tilde{y}%
\left( x,t\right) -i\left( x\right) .
\end{equation*}%
Consequently, we have $\nabla u:=\nabla \tilde{u}\left( x,t\right) =F+I$,
where $I$ is the second-rank unit tensor.

As a matter of notation, we shall write $u_{t}$, $u_{tt}$, and $\nu _{t}$
for the values $\dot{y}:=\frac{d\tilde{y}\left( x,t\right) }{dt}$, $\ddot{y}%
:=\frac{d^{2}\tilde{y}\left( x,t\right) }{dt^{2}}$, and $\dot{\nu}:=\frac{d%
\tilde{\nu}\left( x,t\right) }{dt}$, respectively, the latter chosen for the
sake of notational uniformity. The velocity in the physical space
and the phason time rate just listed are expressed in Lagrangian representation,
i.e. as fields over the reference place and the time scale.
We can have an Eulerian representation of such fields, i.e. we can consider
them defined over the actual shape $\mathcal{B}_{a}=\tilde{y}(\mathcal{B},t)$.
In this case we write $\mathrm{v}(y,t)$ and $\upsilon(y,t)$. Since at $x$ and $t$,
the vector $\dot{y}(x,t)$ is tangent to $\mathcal{B}_{a}$ at the point $y$,
we have the standard identity
\begin{equation*}
\dot{y}(x,t)=\mathrm{v}(y,t).
\end{equation*}
An analogous relation \emph{does not hold} between $\upsilon(y,t)$
and $\dot{\nu}(x,t)$. The lack of identity depends on the circumstance 
that $\upsilon$ is the time rate of the Eulerian representation of the
phason field, which is a map $\tilde{\nu}_{a}$ defined by
\begin{equation*}
\tilde{\nu}_{a}:=\tilde{\nu}\circ\tilde{y}^{-1},
\end{equation*}
a definition possible for $\tilde{y}$ is one-to-one. The subscript $a$ means \emph{actual}, i.e. \emph{referred to the deformed configuration}, here and in what follows.

The condition
\begin{equation*}
|\nabla u|\ll 1
\end{equation*}
defines the \textbf{small strain regime}, in which we develop the analyses
presented in Sections 3 and 4. 
In this setting we can \textquoteleft confuse' $\mathcal{B}$ with $\mathcal{B}_{a}$, 
$u_{t}$ with $\mathrm{v}$, $\nu$ with $\nu_{a}:=\tilde{\nu_{a}}(y,t)$.

\subsection{Changes in observers}

According to the definition proposed explicitly in \cite{M06} and further
refined in \cite{M14}, we define \textbf{observer} \emph{frames of reference
assigned on all spaces necessary to describe the shape of a body and its
motion}. In the setting discussed here, an observer is then (1) a frame in $%
\mathcal{\tilde{E}}^{3}$ or--it is the same--in the pertinent translation
space containing $u$, a space identified with $\mathbb{R}^{3}$, once we fix
an origin, (2) a frame in $\mathcal{E}^{3}$, (3) a frame in $%
\mathcal{V}^{3}$--all identified with copies of $\mathbb{R}^{3}$, which we
can consider differing one another just by the identification--and (4) a
time scale.

We consider time-varying synchronous changes in observers leaving invariant
the reference space and changing the frame $\mathcal{\tilde{E}}^{3}$ by a
rigid body motion. Precisely, let us write $\mathcal{O}$ and $\mathcal{O}%
^{\prime }$ for these two observers. A place $y$ for $\mathcal{O}$ becomes $%
y^{\prime }$ for $\mathcal{O}^{\prime }$, with%
\begin{equation*}
y^{\prime }:=w\left( t\right) +Q\left( t\right) \left( y-y_{0}\right) ,
\end{equation*}%
where $w\left( t\right) $ and $Q\left( t\right) $ are the values at $t$ of
time-differentiable maps $t\longmapsto w\left( t\right) \in \mathbb{R}^{3}$, 
$t\longmapsto Q\left( t\right) \in SO\left( 3\right) $, with $t$ running in
the selected time interval, and $y_{0}$ an arbitrary point in space. The
time rates are then $\dot{y}$ for the first observer and $\dot{y}^{\prime }=%
\dot{w}+\dot{Q}\left( y-y_{0}\right) +\dot{y}$ for the second. By rotating
back by $Q^{-1}=Q^{T}$ the rate $\dot{y}^{\prime }$ into the frame defining $%
\mathcal{O}$, and indicating by $\dot{y}^{\ast }$ the rotated velocity,
which is $Q^{T}\dot{y}^{\prime }$, we get%
\begin{equation*}
\dot{y}^{\ast }:=c+q\times \left( y-y_{0}\right) +\dot{y},
\end{equation*}%
where $c:=Q^{T}w$, and $q$ is the axial vector of the skew-symmetric tensor $%
Q^{T}\dot{Q}$, both depending on time only. Since $\dot{y}=u_{t}$, for the
displacement rate, under the change in observer considered here, we find%
\begin{equation*}
u_{t}^{\ast }=c\left( t\right) +q\left( t\right) \times \left(
y-y_{0}\right) +u_{t}.
\end{equation*}%
Since $\dot{y}(x,t)=v(y,t)$, we can also write
\begin{equation*}
\mathrm{v}^{\ast }:=c\left( t\right) +q\left( t\right) \times \left(
y-y_{0}\right) +\mathrm{v}.
\end{equation*}
The distinction between the ambient space $\mathcal{\tilde{E}}^{3}$ and the
phason one $\mathcal{V}^{3}$ is just matter of modeling. Atomic flips,
determining phason defects, occur in the physical space, indeed. We have
also to remind that the notion of observer is just a formal representation
of the concrete action of recording a phenomenon. When we rotate an observer
in space we should perceive rotated the atomic flips. They are not affected
by rigid translations in space for they are internal degrees of freedom.
Consequently, with $\nu $ the value of the phason field for $\mathcal{O}$, the
observer $\mathcal{O}^{\prime }$ records a value $\nu ^{\prime }=Q\nu $. The
relevant rates are then $\dot{\nu}$ and $\dot{\nu}^{\prime }=Q\dot{\nu}+\dot{%
Q}\nu $ respectively. By writing $\dot{\nu}^{\ast }=\nu _{t}^{\ast }$ for $%
Q^{T}\dot{\nu}^{\prime }$, we get%
\begin{equation*}
\dot{\nu}^{\ast }=\dot{\nu}+q\times \nu .
\end{equation*}

\subsection{External power, invariance and balance}
We have already mentioned in the Introduction that we derive
balance equations from the invariance of power over a generic 
\emph{part} of the body. The word \textbf{part} indicates here a
subset $\mathfrak{b}$
of $\mathcal{B}$ with non-null volume measure and the same regularity 
of $\mathcal{B}$ itself or a subset $\mathfrak{b}_{a}$ of 
the current macroscopic shape $\mathcal{B}_{a}=\tilde{y}(\mathcal{B},t)$ 
of the body. Given a generic $\mathfrak{b}_{a}$, we divide as usual all
actions exerted on $\mathfrak{b}_{a}$ by the environment and the rest of 
the body into bulk and contact families, the latter intended to be exerted 
through the boundary of the part considered. Each family is also subdivided
into standard and phason components, all defined by the expression
of the power that the external action must perform over $\mathfrak{b}_{a}$
to change its state of motion with velocity $\mathrm{v}$ in the 
physical space and phason rate $\upsilon$. For this reason
we call such a power \textbf{external}, indicating it by $\mathcal{P}_{\mathfrak{b}}^{ext}$
and defining it in Eulerian representation by
\begin{equation}
\mathcal{P}_{\mathfrak{b}_{a}}^{ext}\left( \mathrm{v},\upsilon\right) :=\int_{%
\mathfrak{b}_{a}}\left( b^{\ddagger }_{a}\cdot \mathrm{v}+\beta ^{\ddagger }_{a}\cdot \upsilon%
\right) \text{ }d\mu(y)+\int_{\partial \mathfrak{b}_{a}}\left( \mathfrak{t}\cdot 
\mathrm{v}+\tau \cdot \upsilon\right) \text{ }d\mathcal{H}^{2},  \label{ExtPow}
\end{equation}%
were $d\mathcal{H}^{2}$ is the surface measure along $\partial \mathfrak{b}_{a}$ and 
$d\mu(y)$ is the volume measure in $\mathcal{B}_{a}$.

At $y \in \partial\mathfrak{b}_{a}$, where $\partial\mathfrak{b}_{a}$ is 
oriented by the normal $n$, the standard traction $\mathfrak{t}$
depends on $y$ itself and $n$, besides the time $t$ (Cauchy's assumption and 
Hamel-Noll theorem). Here we assume the same
dependence for the phason traction $\tau$, i.e. we impose
\begin{equation*}
\tau:=\tilde{\tau}(y,n),
\end{equation*}
in addition to
\begin{equation*}
\mathfrak{t}:=\tilde{\mathfrak{t}}(y,n),
\end{equation*}
where we leave unexpressed the dependence on time for the sake of conciseness of some formulas below.

What we impose to $\mathcal{P}_{\mathfrak{b}_{a}}^{ext}\left( \mathrm{v},\upsilon\right)$ 
is an axiom of invariance.

\ \

\textbf{Axiom}: $\mathcal{P}_{\mathfrak{b}_{a}}^{ext}\left( \mathrm{v},\upsilon\right)$ 
is invariant under rigid-body-based
changes in observers, i.e.
\begin{equation*}
\mathcal{P}_{\mathfrak{b}_{a}}^{ext}\left( \mathrm{v}^{\ast },\upsilon^{\ast }\right)=\mathcal{P}_{\mathfrak{b}_{a}}^{ext}\left( \mathrm{v},\upsilon\right)
\end{equation*}
for any choice of $c$ and $q$.

\ \

\begin{theorem}
The axiom of invariance implies the following list of assertions:
\begin{itemize}
\item[(a)] If the fields $y\longmapsto b^{\ddagger }_{a}$, $y\longmapsto \nu\times\beta^{\ddagger } $, $y \longmapsto \mathfrak{t}$ and $y\longmapsto \tau$ are integrable over $\mathcal{B}_{a}$, the following integral balances hold for \textbf{any} part $\mathfrak{b}_{a}$ of $\mathcal{B}_{a}$ and for $\mathcal{B}_{a}$
itself:
\begin{equation}
\int_{\mathfrak{b}_{a}}b_{a}^{\ddagger }\text{ }d\mu(y)+\int_{\partial \mathfrak{b}_{a}}%
\mathfrak{t}\text{ }d\mathcal{H}^{2}=0,  \label{IBF}
\end{equation}%
\begin{equation}
\int_{\mathfrak{b}_{a}}\left( \left( y-y_{0}\right) \times b_{a}^{\ddagger }+\nu
\times \beta_{a} ^{\ddagger }\right) \text{ }d\mu(y)+\int_{\partial \mathfrak{b}_{a}%
}\left( \left( y-y_{0}\right) \times \mathfrak{t}+\nu \times \tau \right) 
\text{ }d\mathcal{H}^{2}=0.  \label{IBC}
\end{equation}
\item[(b)] If the standard traction is continuous with respect to $y$ and the standard bulk action is bounded over $\mathcal{B}_{a}$ at every instant, $\mathfrak{t}$ satisfies the action-reaction principle
\begin{equation}
\mathfrak{t}(y,n)=-\mathfrak{t}(y,-n). \label{a-r}
\end{equation}
\item[(c)] In the same continuity conditions, a second-rank tensor $\sigma$ independent of $n$ exists and is
such that 
\begin{equation}
\mathfrak{t}(y,n)=\sigma(y)n(y). \label{Cauchy}
\end{equation}
\item[(d)] If the phason traction is continuous with respect to $y$ and the field $y\longmapsto \nu\times\beta^{\ddagger}$ is bounded over $\mathcal{B}_{a}$ at every instant, $\tau$ satisfies a non-standard action-reaction principle
\begin{equation}
\nu_{a}(y) \times(\tau(y,n)-\tau(y,-n)=0. \label{a-r-phason}
\end{equation}
\item[(e)] In the same regularity conditions above, a second-rank
tensor $\mathcal{S}_{a}$ independent of $n$ exists and is such that
\begin{equation}
\tau(y,n)=\mathcal{S}_{a}(y)n(y), \label{phason-Cauchy}
\end{equation}
a tensor that we call \textbf{phason stress}.
\item[(f)] If the field $y \longmapsto \sigma(y)$ is  $C^{1}$ over $\mathcal{B}_{a}$ and just continuous over its boundary, equation
(\ref{IBF}) implies the validity of the standard pointwise balance of forces
\begin{equation}
b_{a}^{\ddagger }+\mathrm{div}\sigma=0. \label{pointCau}
\end{equation}
\item[(g)] If, in addition, the field $y \longmapsto \mathcal{S}_{a}(y)$ is  $C^{1}$ over $\mathcal{B}_{a}$ and just continuous over its boundary, equation (\ref{IBC}) implies the existence of a vector $z_{a}$ such that
\begin{equation}
\mathrm{div}\mathcal{S}_{a}+\beta_{a} ^{\ddag }-z_{a}=0  \label{f2}
\end{equation}
and
\begin{equation}
\mathrm{Skw}\sigma =\frac{1}{2}\mathrm{e}(\nu _{a}\times z+(\mathrm{e}\nabla_{y}\left( 
\nu _{a}\right) )^{T }\mathcal{S}_{a}),  \label{f3}
\end{equation}
where the apex T means transposition, $\mathrm{e}$ is Ricci's alternating symbol, and $\nabla_{y}$ is the gradient with respect to $y$.
\item[(h)] The external power satisfies the following relation:
\begin{equation}
\mathcal{P}_{\mathfrak{b}_{a}}^{ext}\left( \mathrm{v},\upsilon\right)=\int_{\mathfrak{b}_{a}}(\sigma\cdot\nabla_{y}\mathrm{v}+z_{a}\cdot\upsilon+\mathcal{S}_{a}\cdot\nabla_{y}\upsilon)\: d\mu(y), \label{plv}
\end{equation}
for any choice of the rates involved. The right-hand side term takes the name of \textbf{inner power}.
\end{itemize}
\end{theorem}

\begin{proof}
The first item follows trivially by the arbitrariness of $c$ and $q$ for the
axiom implies
\begin{equation}
\mathcal{P}_{\mathfrak{b}_{a}}^{ext}\left( c+q\times (y-y_{0}),q\times \nu\right)=0.
\end{equation}
The first integral balance (\ref{IBF}) implies the boundedness of the
absolute value of the 
traction average over the boundary of any part, once the bulk actions are 
bounded. Consequently, standard arguments (see \cite{T}) allow us
to obtain the action-reaction principle (\ref{a-r}) for $\mathfrak{t}$ and the 
Cauchy theorem (\ref{Cauchy}) about the existence of a stress independent of $n$. 
Notice that the existence of the stress tensor can be obtained in less stringent regularity 
assumptions (see, e.g., \cite{Seg1} and \cite{Sil}).

On defining $r:=\left( y-y_{0}\right) \times
b_{a}^{\ddagger }+\nu \times \beta _{a}^{\ddagger }$ and $p:=\left(
y-y_{0}\right) \times \mathfrak{t}+\nu \times \tau $, the integral equation %
\ref{IBC} writes obviously as%
\begin{equation*}
\int_{\mathfrak{b}_{a}}r\text{ }d\mu (y)+\int_{\partial \mathfrak{b}_{a}}p%
\text{ }d\mathcal{H}^{2}=0.
\end{equation*}%
Since $y_{0}$ is arbitrary, once we choose $y$, we can select $y_{0}$ such
that the above boundedness assumption about $\left\vert b_{a}^{\ddagger
}\right\vert $ over $\mathcal{B}$ may imply the one of $\left(
y-y_{0}\right) \times b_{a}^{\ddagger }$. Moreover, the above assumption of
the boundedness $\left\vert \nu \times \beta _{a}^{\ddagger }\right\vert $
over $\mathcal{B}$ implies the one of $r$. The assumed boundedness of the
first integral implies the one of the absolute value of the right-hand side
term, so that we can apply the standard procedure adopted for the traction $%
\mathfrak{t}$ (see \cite{T}), obtaining a non-standard action-reaction
principle (\ref{a-r-phason}). We find, in fact,%
\begin{equation*}
p\left( x,n\right) =-p\left( x,-n\right) ,
\end{equation*}%
i.e.%
\begin{equation*}
\left( y-y_{0}\right) \times \left( \mathfrak{t}\left( x,n\right) -\mathfrak{%
t}\left( x,-n\right) \right) +\nu \times \left( \tau \left( x,n\right) -\tau
\left( x,-n\right) \right) =0,
\end{equation*}%
so that from (\ref{a-r}) we find (\ref{a-r-phason}).

Not writing explicitly time for the sake of conciseness, we can use a
tetrahedron type argument to show the linearity of $p$ with respect to $n$,
namely we show the existence of a second-rank tensor $A\left( x\right) $
such that%
\begin{equation*}
p\left( x,n\right) =A\left( x\right) n\left( x\right) .
\end{equation*}%
Then we find%
\begin{equation*}
\left( y\left( x\right) -y_{0}\right) \times \mathfrak{t}+\nu \left(
x\right) \times \tau =\left( y\left( x\right) -y_{0}\right) \times P\left(
x\right) n\left( x\right) +\nu \left( x\right) \times \tau =A\left( x\right)
n\left( x\right) ,
\end{equation*}%
so that%
\begin{equation*}
\nu \left( x\right) \times \tau \left( x,n\right) =A\left( x\right) n\left(
x\right) -\left( y\left( x\right) -y_{0}\right) \times P\left( x\right)
n\left( x\right) ,
\end{equation*}%
which implies the linearity of $\tau \left( x,n\right) $ with respect to $n$.

The localization of the integral balance of forces, namely equation
(\ref{IBF}), which is possible due to the arbitrariness of $\mathfrak{b}_{a}$, 
gives rise to the local balance (\ref{pointCau}).
In contrast, the localization of the integral balance (\ref{IBC}) implies
\begin{equation*}
\nu _{a}\times (\mathrm{div}\mathcal{S}_{a}+\beta_{a} ^{\ddag })=\mathrm{e}\sigma^{T}-(\nabla_{y}
\nu _{a})^{T }\mathcal{S}_{a},
\end{equation*}
which indicates the existence of a vector, say $z_{a}$, satisfying the
equation (\ref{f3}) with the constraint (\ref{f2}).
\end{proof}

\subsection{Inertia terms}
A standard assumption is the additive decomposition of the bulk forces
$b^{\ddagger}$ into inertial, $b^{in}$, and non-inertial, $b$
components, the former identified by definition by equating their
power to the negative of the time rate of the (canonical) kinetic energy.
Since we do not have in mind any external bulk direct action over the phason flips, except perhaps the possible influence of magnetic fields in magnetizable quasicrystals, a material class not treated here, according to the proposal in \cite{MP13}, we consider $\beta^{\ddagger}$ just with inertial nature.

Our inertia axiom is then the presumed validity of the integral
equation
\begin{equation*}
\text{rate of the kinetic energy of} \text{ }\mathfrak{b}_{a}= -\int_{\mathfrak{b}_{a}}(b_{a}^{in}\cdot\mathrm{v}+\beta_{a}^{\ddagger}\cdot\upsilon)\: d\mu(y),
\end{equation*}
for any choice of the rates involved and any part $\mathfrak{b}_{a}$.\footnote{Notice that these requirements are superabundant for we
could just considered the integral extended to the whole body
macroscopic shape $\mathcal{B}_{a}$, the arbitrariness of the rate fields allowing the selection of parts through the possibility of choosing
compactly supported rate fields. However, we maintain the
superabundant choice because it seems to us that it puts better in
evidence the physical nature of the requirement, which would fail in the
relativistic setting for the arbitrariness of $\mathfrak{b}_{a}$ could be
maintained, while the rate fields could not be selected at will.}
The key point is then the expression of the kinetic energy.
There is a debate about the possible existence of a peculiar phason kinetic energy.
On one side, who is interested in having
a structure duplicating the standard elasticity in a higher-dimensional space,
with the obvious analytical advantages, would hope for it.
However, just
three sound-like branches seem to appear in dynamic spectra recorded in
experiments (see, e.g., \cite{SvS}) so that we should be inclined
in not considering phason inertia.
For this reason, we write explicitly the previous balance as
\begin{equation}
\frac{d}{dt}\int_{\mathfrak{b}_{a}}\frac{1}{2}\rho|v|^{2}\text{ }d\mu (y)= -\int_{\mathfrak{b}_{a}}(b_{a}^{in}\cdot\mathrm{v}+\beta_{a}^{\ddagger}\cdot\upsilon)\text{ }d\mu (y), \label{In}
\end{equation}
where $\rho:=\tilde{\rho}(y,t)$ is the value at $y$ and $t$ of the mass density, assumed to be differentiable with respect to its entries, in the actual shape of the body. Here we presume
that $\rho$ is conserved along the motion, i.e. it satisfies the local mass balance
\begin{equation*}
\frac{\partial\rho}{\partial t}+\mathrm{div}(\rho\mathrm{v})=0.
\end{equation*}
By taking into account the previous equation and applying a standard
transport theorem to compute the time derivative of the first integral
in the equation (\ref{In}), in which the integration domain depends on time, the arbitrariness of $\mathfrak{b}_{a}$--or the one of the rate
fields, which is the same--implies
\begin{equation*}
b_{a}^{in}=-\rho a,
\end{equation*}
with $a:=\dot{\mathrm{v}}$ the \textbf{acceleration} in Euclidean
representation, and
\begin{equation*}
\beta_{a}^{\ddagger}\cdot\upsilon=0.
\end{equation*}
This last identity implies that $\beta_{a}^{\ddagger}$ must be of the form
$\beta_{a}^{\ddagger}=h \times\nu$, with $h$ a generic vector.
P. M. Mariano and J. Planas suggested \cite{MP13} to identify the vector $h$
with $-\mathrm{curl}\text{ }\mathrm{v}$ so that we have
\begin{equation*}
\beta_{a}^{\ddagger}=-(\mathrm{curl}\text{ }\mathrm{v})\times\upsilon.
\end{equation*}
Such a choice is motivated by the presumption that the local deformation
spin tends to rotate the lattice structures influencing the way the
atomic flips may develop. At continuum scale such circumstance should generate a coupled 
gyroscopic effect, as represented above.

\subsection{Constitutive structures}
Constitutive restriction on the dependence of the stresses and the 
phason self-action on the state variables characterizing a material class
are prescribed by the need of not violating the second-law of
thermodynamics. The statement remains vague till we specify an
expression of the second law.
In large strain regime, in which we
distinguish between reference and actual shapes, it is natural to
write such an expression in referential form. For it we write in isothermal setting
\begin{equation*}
\frac{d}{dt}\int_{\mathfrak{b}}\psi\text{ }d\mu(x)-\mathcal{P}^{ext}_{\mathfrak{b}}(\dot{y},\dot{\nu})\leq 0,
\end{equation*}
adapting to the description of quasicrystals the traditional viewpoint
in continuum mechanics on the constitutive matter  (see \cite{C-N} for it).
In the previous
inequality $\psi$ is the \textbf{free energy density} and
$\mathcal{P}^{ext}_{\mathfrak{b}}(\dot{y},\dot{\nu})$ the
referential description of the external power, obtained by changing
variables in the integrals. It reads
\begin{equation*}
\mathcal{P}^{ext}_{\mathfrak{b}}(\dot{y},\dot{\nu}):=
\int_{\mathfrak{b}}(b^{\ddagger}\cdot\dot{y}+\beta^{\ddagger}\cdot\dot{\nu})\: d\mu(x)+\int_{\partial\mathfrak{b}}(\mathfrak{t}\cdot\dot{y}+\tau\cdot\dot{\nu})\: d\mathcal{H}^{2},
\end{equation*}
where $b^{\ddagger}:=(\mathrm{det}F)b_{a}^{\ddagger}$,
$\beta^{\ddagger}:=(\mathrm{det}F)\beta_{a}^{\ddagger}$,
$\mathfrak{t}$ and $\tau$ are considered as the values of fields
defined over $\mathcal{B}$ through $\tilde{\mathfrak{t}}(y(x,t),t)$
and $\tilde{\tau}(y(x,t),t)$. With this version of the external power,
the invariance axiom above requires the identity
$\mathcal{P}^{ext}_{\mathfrak{b}}(\dot{y},\dot{\nu})=
\mathcal{P}^{ext}_{\mathfrak{b}}(\dot{y}^{*},\dot{\nu}^{*})$
for any choice of $c$, $q$, involved in the definitions of $\dot{y}^{*}$
and $\dot{\nu}^{*}$, and the part considered. By exploiting such identity
we can prove the referential version of Theorem 2.1, which includes, in
particular, the relation
\begin{equation*}
\mathcal{P}^{ext}_{\mathfrak{b}}(\dot{y},\dot{\nu}):=
\int_{\mathfrak{b}}(P\cdot\dot{F}+z\cdot \dot{\nu}+\mathcal{S}\cdot\dot{N})\: d\mu(x)
\end{equation*}
to be substituted into the mechanical dissipation inequality. In the
previous identity, $P$ is the standard first Piola-Kirchhoff stress
tensor defined by $P:=(\mathrm{det}F)\sigma F^{-\mathrm{T}}$,
$z$ the referential phason self-action $z:=(\mathrm{det}F)z_{a}$
and $\mathcal{S}$ the referential phason stress (or microstress if you
want to accept a nomenclature more common to the general model-building framework of the mechanics of complex materials,
which includes the quasicrystal modeling) $\mathcal{S}:= (\mathrm{det}F)\mathcal{S}_{a} F^{-\mathrm{T}}$.
\begin{itemize}
\item When we presume that the free energy $\psi$, the stresses $P$
and $\mathcal{S}$, and the self-action $z$ depend all on $F$, $N$ and $\nu$, besides $x$, with $\psi$ a differentiable function of its entries, the arbitrariness of the rates in the mechanical dissipation inequality implies
the constitutive restrictions
\begin{equation*}
P=\frac{\partial \psi}{\partial F},\:\: \mathcal{S}=\frac{\partial \psi}{\partial N},\:\:z=\frac{\partial \psi}{\partial \nu},
\end{equation*}
i.e.
\begin{equation*}
\sigma=(\mathrm{det}F)^{-1}\frac{\partial \psi}{\partial F}F^{\mathrm{T}},\:\: \mathcal{S}_{a}=(\mathrm{det}F)^{-1}\frac{\partial \psi}{\partial N}F^{\mathrm{T}},\:\:z_{a}=(\mathrm{det}F)^{-1}\frac{\partial \psi}{\partial \nu}.
\end{equation*}
They characterize the elastic setting for quasicrystals.
\item By fixing $\nu$ and $N$, a standard argument shows that objectivity for $\psi$, i.e. invariance under the action of $SO(3)$ on the physical space, and convexity of $\psi$ with respect to $F$ are physically incompatible. Consequently, we commonly accept a polyconvex dependence of $\psi$ on $F$. With respect to $N$, the free energy can be quadratic in the so-called \emph{phason locked phase}, and $\psi$ may admit a decomposed Ginzburg-Landau-type structure.
In the so-called \emph{phason unlocked phase}, $\psi$ depends on $|N|$.
With reference to the phason locked phase,
the existence of ground states (minimizers of the energy) has been found in \cite{MM} as a special case of a more general result presented there, further generalized in \cite{FMS}.
\item In small strain regime the dependence of the energy
can be quadratic. With reference to the homogeneous and isotropic case, with $\varepsilon:=\mathrm{Sym}\nabla u$ the small strain tensor and $I$
the second-rank unit tensor, a rather general expression of the energy
has been derived in \cite{MP13}; it reads
\begin{equation}
\begin{aligned}
\psi =&\frac{1}{2}\lambda \left( \varepsilon \cdot I\right) ^{2}+\mu
\varepsilon \cdot \varepsilon  \\ 
&+\frac{1}{2}k_{1}\left( N\cdot I\right) ^{2}+k_{2}\mathrm{Sym}N\cdot 
\mathrm{Sym}N+k_{2}^{\prime }\mathrm{Skw}N\cdot \mathrm{Skw}N \\
&+k_{3}\left( \varepsilon \cdot I\right) \left( N\cdot I\right)
+k_{3}^{\prime }\mathrm{Sym}N\cdot \varepsilon  \\
&+\frac{1}{2}k_{0}\left\vert \nu \right\vert ^{2} \label{energy}
\end{aligned}
\end{equation}
from which we get%
\begin{equation}
\sigma =\lambda \left( \mathrm{tr}\varepsilon \right) I+2\mu \varepsilon
+k_{3}\left( \mathrm{tr}N\right) I+k_{3}^{\prime }\mathrm{Sym}N, \label{sigma}
\end{equation}%
\begin{equation}
z_{a}=k_{0}\nu , \label{zeta}
\end{equation}%
\begin{equation}
\mathcal{S}_{a}=k_{1}\left( \mathrm{tr}N\right) I+2k_{2}\mathrm{Sym}%
N+2k_{2}^{\prime }\mathrm{Skw}N+k_{3}\left( \mathrm{tr}\varepsilon \right)
I+k_{3}^{\prime }\varepsilon . \label{microstress}
\end{equation}%
$\lambda $ and $\mu $ are standard Lam\'{e} constants. The others are
elastic constants related with the phason field. Experiments inform us
about the values of $k_{i}$, $i=1,2,3$, with some fluctuations in the literature (see, e.g., \cite{Amazit}, \cite{Letoublon}, \cite{Richer}, \cite{Walz}), but we do not know $k'_{2}$, $k'_{3}$ and $k_{0}$.
\item The quadratic expression of the energy written above is a special
case of
\begin{equation*}
\psi\left( \nabla u,\nu ,\nabla\nu \right) =\frac{1}{2}\nabla u\cdot \mathbb{C}\nabla u+\nabla u\cdot 
\mathbb{K}^{\prime }\nabla\nu +\frac{1}{2}\nabla\nu \cdot \mathbb{K}\nabla\nu +\frac{1}{2}%
k_{0}\left\vert \nu \right\vert ^{2},
\end{equation*}
with $\mathbb{C}$, $\mathbb{K}$ and $\mathbb{K}'$ constitutive fourth-rank tensors, endowed at least with major symmetries, but
 \emph{it is a bit more general then} the common choice
\begin{equation*}
\mathbb{C}_{ijhk}=\lambda \delta _{ij}\delta _{hk}+\mu \left( \delta
_{ih}\delta _{jk}+\delta _{ik}\delta _{jh}\right) ,
\end{equation*}%
\begin{equation*}
\mathbb{K}_{ijhk}^{\prime }=k_{1}\delta _{ih}\delta _{jk}+k_{2}\left( \delta
_{ij}\delta _{hk}-\delta _{ik}\delta _{jh}\right) ,
\end{equation*}%
\begin{equation*}
\mathbb{K}_{ijhk}=k_{3}\left( \delta _{i1}-\delta _{i2}\right) \left( \delta
_{ij}\delta _{hk}-\delta _{ih}\delta _{jk}+\delta _{ik}\delta _{jh}\right) 
\end{equation*}%
(see, e.g., \cite{HWD}), where $i,j,h,k=1,2,3 $, $\delta _{ij}$ is the Kronecker symbol and the constants satisfy the inequalities
\begin{eqnarray*}
\mu  &>&0,\text{ \ \ }\lambda +\mu >0,\text{ \ \ }k_{1}>0, \\
k_{1} &>&\left\vert k_{2}\right\vert ,\text{ \ \ }\left\vert
k_{3}\right\vert <\sqrt{\frac{1}{2}\mu \left( k_{1}+k_{2}\right) },\text{ }\:%
k_{0}\geq 0,
\end{eqnarray*}
allowing nonnegative definition of the energy.
\item We could also imagine to have viscous effects represented through
the dependence of the stresses $P$, $\mathcal{S}$ and the self-action 
$z$ on $\dot{F}$, $\dot{N}$ and $\dot{\nu}$, besides $F$, $N$ and 
$\nu$. The mechanical dissipation inequality excludes the dependence
of $\psi$ on $\dot{F}$, $\dot{N}$ and $\dot{\nu}$, provided that $F$, $N$ and $\nu$ are twice differentiable in time. Consequently,
to be compatible with the second law of thermodynamics, $P$, $\mathcal{S}$ and $z$ must admit a constitutive dependence on the rates of the state variables of the form
\begin{equation*}
P=\tilde{P}(F,N,\nu,\dot{F},\dot{N}, \dot{\nu})=\tilde{P}^{e}(F,N,\nu)+\tilde{P}^{d}(F,N,\nu,\dot{F},\dot{N}, \dot{\nu}),
\end{equation*}
\begin{equation*}
\mathcal{S}=\tilde{\mathcal{S}}(F,N,\nu,\dot{F},\dot{N}, \dot{\nu})=\tilde{\mathcal{S}}^{e}(F,N,\nu)+\tilde{\mathcal{S}}^{d}(F,N,\nu,\dot{F},\dot{N}, \dot{\nu}),
\end{equation*}
\begin{equation*}
z=\tilde{z}(F,N,\nu,\dot{F},\dot{N}, \dot{\nu})=\tilde{z}^{e}(F,N,\nu)+\tilde{z}^{d}(F,N,\nu,\dot{F},\dot{N}, \dot{\nu}),
\end{equation*}
where the superscripts $e$ and $d$ indicate the energetic and the dissipative (viscous) components. By inserting these choices in the mechanical dissipation inequality, we find the dependence of the energetic components of $P$, $\mathcal{S}$ and $z$ form the derivatives of the free energy, as recalled above, with the consequent
expressions of $\sigma^{e}$, $\mathcal{S}^{e}_{a}$ and $z^{e}_{a}$, and the
reduced dissipation inequality
\begin{equation*}
P^{d}\cdot \dot{F}+z^{d}\cdot\dot{\nu}+\mathcal{S}^{d}\cdot \dot{N}\geq 0
\end{equation*}
valid for any choice of the velocity fields, which implies that $P^{d}$, $\mathcal{S}^{d}$ and $z^{d}$ can be considered linear functions of
$\dot{F}$, $\dot{N}$ and $\dot{\nu}$ as their actual counterparts
$\sigma^{d}$, $\mathcal{S}^{d}_{a}$ and $z^{d}_{a}$.

In this case and in small strain setting, for the energetic components of $\sigma$, $\mathcal{S}_{a}$ and $z_{a}$ we
shall consider the energy (\ref{energy}) and dissipative components of the stresses and the self-action given by
\begin{equation*}
\sigma^{d}=\epsilon\nabla u_{t},\:\:\:\mathcal{S}_{a}^{d}=\delta\nabla\nu_{t},
\:\:\:z_{a}=\varsigma\nu_{t},
\end{equation*}
with $\epsilon$, $\delta$ and $\varsigma$ positive constants.
Consequently, the constitutive equations (\ref{sigma}), (\ref{zeta}) and (\ref{microstress}) become
\begin{equation}
\sigma =\lambda \left( \mathrm{tr}\varepsilon \right) I+2\mu \varepsilon
+k_{3}\left( \mathrm{tr}N\right) I+k_{3}^{\prime }\mathrm{Sym}N
+\epsilon\nabla u_{t}, \label{sigma-dis}
\end{equation}%
\begin{equation}
z_{a}=k_{0}\nu + \varsigma\nu_{t}, \label{zeta-dis}
\end{equation}%
\begin{equation}
\mathcal{S}_{a}=k_{1}\left( \mathrm{tr}N\right) I+2k_{2}\mathrm{Sym}%
N+2k_{2}^{\prime }\mathrm{Skw}N+k_{3}\left( \mathrm{tr}\varepsilon \right)
I+k_{3}^{\prime }\varepsilon + \delta\nabla\nu_{t} . \label{microstress-dis}
\end{equation}%
We shall use the constitutive equations (\ref{sigma-dis}) and (\ref*{microstress-dis}) just for technical purposes, due to the regularization induced by the gradients of the rate fields. We
remark here only their mechanical motivation but we do not investigate
further their experimental correspondence for we shall calculate the
limits as $\epsilon$ and $\delta$ tend to zero. In contrast,
there are estimates for $\varsigma$ (see \cite{RL}).
\end{itemize}

\section{Existence results: the linear case}

\subsection{Dynamics with phason diffusion and absence of gyroscopic effects}
In small strain regime and under the validity
of the linear constitutive structures
(\ref{sigma}), (\ref{microstress}) and (\ref{zeta-dis}), in absence of  non-inertial body forces and gyroscopic-type phason inertia,
by imposing $u$ and $\nu$ along $\partial\mathcal{B}$
(Dirichlet boundary conditions) and their values together with those of the velocity $u_{t}$ over $\mathcal{B}$ as initial conditions, the balance equations read
\begin{equation} \label{eq:quasicrystal} 
\begin{aligned}
&\mu \Delta u + \xi \nabla \div u + \kappa \Delta \nu + \bar \xi
      \nabla \div \nu = \rho u_{tt}
      & \textrm{ in } (0, T)\X \mathcal{B},\\[0.1 cm]
     &\zeta \Delta \nu + \gamma \nabla \div \nu + \kappa \Delta u
      +\bar \xi \nabla \div u -\kappa_0 \nu= \varsigma \nu_t
      & \textrm{ in } (0, T)\X \mathcal{B},\\[0.1 cm]
      &u(t, x) = \bar u(x),\,\, \nu(t, x) = \bar \nu (x),
      & \textrm{ on } (0, T)\X\D \mathcal{B}, \\[0.1 cm]
      &u|_{t=0}=u_0, \,\, u_t|_{t=0}=\dot u_0,\,\, \nu|_{t=0}=\nu_0, 
      & \textrm{ on } \mathcal{B},
    \end{aligned}
\end{equation}
where $u_0$, $\dot u_0$ and $\nu_0$ are the initial
data, and the constitutive parameters are constants and satisfy the following relations: $
\xi = \lambda + \mu$, $\bar \xi =k_3 + \frac{1}{2} k'_3$, $\zeta =
k_2 + k'_2$, $\gamma = k_1 + k_2 - k'_2$,  $\kappa =\frac{1}{2} k'_3$, and $\lambda$, $\mu$,
$k_i$, $k'_i$, $i=1,2,3$.

\subsection{Preliminaries and notations}
For
$p\geq 1$, by $L^p(\mathcal{B})$ we indicate the usual Lebesgue space with
norm $\|\cdot\|_p$. For $L^2$ we use the
notation $\|\cdot \|=\|\cdot \|_2$. Moreover, by $W^{k,p}(\mathcal{B})$,
$k$ a non-negative integer and $p$ as above, we denote the usual
Sobolev space with norm $\|\cdot\|_{k,p}$. We write
$W^{1,p}_0(\mathcal{B})$ for the closure of $C_0^{\infty}(\mathcal{B})$ in
$W^{1,p}(\mathcal{B})$ and $W^{-1, p'}(\mathcal{B})$, $p' = p/(p -1)$, for the
dual of $W^{1, p}(\mathcal{B})$ with norm $\|\cdot\|_{-1, p'}$.  Let $X$ be
a real Banach space with norm $\|\cdot\|_X$. We shall use the customary
spaces $W^{k,p}(0, T;X)$, with norm denoted by
$\|\cdot\|_{W^{k,p}(0,T;X)}$, recalling that $W^{0,p}(0, T;X)=L^{p}(0,
T;X)$ are the standard Bochner spaces.  The symbol $\langle\,
\cdot\,,\, \cdot\,\rangle$ indicates as usual the duality pairing.  Here and
in the sequel, we denote by $c$ or $\bar c$ positive constants that
may assume different values, even in the same equation. We also define
\begin{align*}
  \mathcal{H}^1 :=\big\{ v\in W^{1,2}(\mathcal{B})\, :\,
  v_{\vert_{\D\mathcal{B}}}= 0\big\},
\end{align*}
with dual space $\mathcal{H}^{-1}$. We denote by $\mathcal{B}_T$ the
set product $(0, T)\X \mathcal{B}$ and, similarly, with $\D\mathcal{B}_T$ we
indicate $(0, T)\X \D\mathcal{B}$.

\subsection{Existence and uniqueness of weak solutions to \eqref{eq:quasicrystal}}

\begin{defin}[Weak solution] \label{def:weak-reg} We affirm that
  a pair $(u, \nu)$ is a \textquotedblleft weak solution'' to the system
  \eqref{eq:quasicrystal} if, for a given $T>0$, the following conditions hold true: \vspace{-0.3 cm}
  \begin{align}
    \intertext{Regularity:} & 
\begin{aligned}
 &u \in L^\infty(0, T ; \mathcal{H}^1)
    \cap C([0, T ]; L^2(\mathcal{B}))  \cap  C_{weak}([0, T ];
    \mathcal{H}^1),\\
 &\nu \in L^2(0, T ; \mathcal{H}^1) \cap C([0, T ]; L^2(\mathcal{B})),
  \end{aligned} \label{eq:u-nu-reg}\\
    & \begin{aligned}
& u_t \in L^\infty(0, T ; L^2(\mathcal{B})) \cap C_{weak}([0, T ]; L^2(\mathcal{B})),  \,\,
    u_{tt} \in L^2(0, T ; \mathcal{H}^{-1}),\\
    &\nu_t \in L^2(0, T ; L^2(\mathcal{B})).
  \end{aligned}
  \label{eq:Dt-u-nu-reg}\\
    \intertext{Weak formulation: For all $(w, h)\in C_0^\infty(0, T;
      \mathcal{H}^1)\X C^\infty_0(0, T; \mathcal{H}^1)$,}
    & \begin{aligned} \label{eq:weak-form-1} \rho &\int_0^T
      \int_{\mathcal{B}} u_{tt}\, \cdot w + \mu \int_0^T \int_{\mathcal{B}}
      \nabla u \cdot \nabla w + \kappa \int_0^T \int_{\mathcal{B}} \nabla
      \nu \cdot \nabla w
      \\
      &= \int_0^T \int_{\D\mathcal{B}} w\cdot\big(\mu\frac{\D u}{\D n} +
      \kappa\frac{\D \nu}{\D n}\big) + \xi \int_0^T \int_{\mathcal{B}}
      \nabla (\div u) \cdot w + \bar \xi \int_0^T \int_{\mathcal{B}} \nabla
      (\div \nu) \cdot w \!\!\!
    \end{aligned}\\[2 mm]
    & \begin{aligned} \label{eq:weak-form-2} \int_0^T & \int_{\mathcal{B}}
      (\varsigma \nu_t + \kappa_0\nu)\cdot h + \zeta \int_0^T
      \int_{\mathcal{B}} \nabla \nu \cdot \nabla h +
      \kappa \int_0^T \int_{\mathcal{B}} \nabla u \cdot \nabla h \\
      &= \int_0^T \int_{\D\mathcal{B}} h \cdot \big(\kappa\frac{\D u}{\D n}
      + \zeta\frac{\D \nu}{\D n}\big) + \gamma \int_0^T \int_{\mathcal{B}}
      \nabla (\div \nu) \cdot h + \bar \xi \int_0^T \int_{\mathcal{B}} \nabla
      (\div u) \cdot h \!\!\!
    \end{aligned}
  \end{align}
\end{defin}

\noindent where, in order to keep the notation concise, we have erased
all volume, surface and time measures from the space-time integrals
above, a choice that we adopt for the remainder of the paper.

\noindent To prove our existence result, we use the Galerkin method to
approximate a regular weak solution to \eqref{eq:quasicrystal} with
finite dimensional displacement and phason vector fields. This is a
classical argument. Details can be found, e.g., in
\cite{Lions:1969}.  

\noindent Let us consider the set $\{\omega_k \}_{k\in \N}$
of eigenfunctions, with corresponding eigenvalues $\{\lambda_k
\}$, of the problem
\begin{equation*}
  \left. \begin{array}{ll}
      -\mu \Delta u = \lambda u & \textrm{ in } \mathcal{B},\\
      u=\bar u &  \textrm{ on } \D\mathcal{B},
    \end{array} \right.
\end{equation*}
we define $X_m := \spann\{\omega_1,\dots, \omega_m\}$ and indicate by $P_m$ the
orthogonal projection operator from $\mathcal{H}^1$ over
$X_m$. Similarly, we also introduce the set $\{\vartheta_r \}_{r\in
  \N}$ of the eigenfunctions, with corresponding eigenvalues
$\{\varpi_r \}$, of
\begin{equation*}
  \left. \begin{array}{ll}
      -\zeta \Delta \nu + \kappa_0\nu = \varpi \nu & \textrm{ in } \mathcal{B},\\
      \nu=\bar \nu &  \textrm{ on } \D\mathcal{B}.
    \end{array} \right.
\end{equation*}
We define $Y_n :=
\spann\{\vartheta_1,\dots, \vartheta_n\}$ and indicate by $\Pi_n $ the
orthogonal projection from $\mathcal{H}^1$ over $Y_n$.

\noindent We are looking to approximate functions
\begin{equation}
  u^m(t, x) =\sum_{i=1}^m d_i^m(t)\omega_i(x)\,\, \textrm{ and }\,\,
  \nu^m(t, x) =\sum_{j=1}^m e_j^m(t)\vartheta_j(x), 
\end{equation}
which are solutions of the system of ODEs below, for all $( \omega_k,
\vartheta_r) \in X_m\X Y_m$, $1\leq k \leq m$, $1\leq r \leq m$, and
$t \in [0, T ]$:
\begin{align*}
  & \begin{aligned} \rho \int_{\mathcal{B}} & u^m_{tt}\, \cdot \omega_k +
    \mu \int_{\mathcal{B}} \nabla u^m \cdot \nabla \omega_k + \kappa
    \int_{\mathcal{B}} \nabla \nu^m \cdot \nabla \omega_k
    \\
    &= \int_{\D\mathcal{B}} \big(\mu^m\frac{\D u^m}{\D n} + \kappa\frac{\D
      \nu^m}{\D n}\big)\cdot \omega_k + \xi \int_{\mathcal{B}} \nabla (\div
    u^m )\cdot \omega_k + \bar \xi \int_{\mathcal{B}} \nabla (\div \nu^m)
    \cdot \omega_k,
  \end{aligned}\\[2 mm]
  & \begin{aligned} \int_{\mathcal{B}} & (\varsigma \nu^m_t +
    \kappa_0\nu^m)\cdot \vartheta_r + \zeta \int_{\mathcal{B}} \nabla \nu^m
    \cdot \nabla \vartheta_r +
    \kappa  \int_{\mathcal{B}} \nabla u^m \cdot \nabla \vartheta_r \\
    &= \int_{\D\mathcal{B}} \vartheta_r \cdot \big(\kappa\frac{\D u^m}{\D
      n} + \zeta\frac{\D \nu^m}{\D n}\big) + \gamma \int_{\mathcal{B}}
    \nabla (\div \nu^m) \cdot \vartheta_r + \bar \xi \int_{\mathcal{B}}
    \nabla (\div u^m) \cdot \vartheta_r.
  \end{aligned}
\end{align*}

\noindent As a consequence, we have the following inclusions: 
$u^m \in L^2(0, T ; X_m )$, $\nu^m \in L^2(0, T; Y_m)$,
$u^m_t \in L^2(0, T ; X_m )$ and $\nu_t^m\in L^2(0, T; Y_m)$. 
The Sobolev embedding theorem for functions (of a single variable
$t$) implies $u_m \in C([0, T ]; X_m )$ and $\nu_m \in C([0, T ];
Y_m )$, so the initial conditions $u_m (0) = P_m u_0$ and
$\nu_m(0)=\Pi_m \nu_0$ make sense.\smallskip

\noindent The Galerkin approximation procedure, combined with a compactness
argument (the Aubin-Lions lemma) and suitable a priori estimates, implies a first result.

\begin{theorem} \label{thm:linear-case} Assume $\mu >  -\lambda$,
$\kappa >0$, $\bar \xi >0$, 
$ \mu, \zeta > 2\kappa$, and 
$\xi, \gamma > 2\bar{\xi}$.
  Assume also $u_0, \nu_0\in W^{1,2}(\mathcal{B})$ so that $\nabla
  u(0, x)= \nabla u_0$ and $\nabla \nu (0, x) = \nabla \nu_0 (x)$ on
  $\mathcal{B}$ 
  and $\bar u, \bar \nu \in L^2(\D\mathcal{B})$.
  Then, a
  unique regular weak solution to the problem \eqref{eq:quasicrystal} exists.
\end{theorem}
\begin{proof}
  We proceed formally (since we lack the needed regularity to test
  directly against $(u, \nu)$ or $(u_t, \nu_t)$), but the procedure
  actually goes through the use of the Galerkin approximation
  functions $(u^m, \nu^m)$. Thus, to keep the notation compact we shall
  use $(u, \nu)$ in place of $( u^m, \nu^m)$, reminding however to adopt
  the sequence of the Galerkin approximations when we extract
  a suitable convergent subsequence.

  \noindent By multiplying first and second equations in
  \eqref{eq:quasicrystal} respectively by $u_{t}$ and $\nu_{t}$ in
  $L^2(\mathcal{B})$, by means of standard calculations we infer that
\begin{align*}
            &\begin{aligned}
          \frac{\rho}{2}\frac{d}{dt}\| u_t \|^2 +
          \frac{\mu}{2}\frac{d}{dt}\|\nabla u\|^2 = \mu
          \int_{\D\mathcal{B}} u_t\cdot \frac{\D u}{\D n} + \xi \int_\mathcal{B}
          \nabla (\div u) \cdot u_t
          + & \kappa  \int_\mathcal{B} \Delta \nu \cdot u_t\\
          & + \bar \xi \int_\mathcal{B} \nabla (\div \nu) \cdot u_t,
        \end{aligned}\\[1.5 em]
        &\begin{aligned} \varsigma \|\nu_t\|^2 +
          \frac{\kappa_0}{2}\frac{d}{dt}\|\nu\|^2 +
          \frac{\zeta}{2}\frac{d}{dt}\|\nabla \nu\|^2 = \zeta
          \int_{\D\mathcal{B}} \nu_t\cdot \frac{\D \nu}{\D n} + \gamma
          \int_\mathcal{B} \nabla (\div &\nu) \cdot
          \nu_t +  \kappa \int_\mathcal{B} \Delta u  \cdot \nu_t\\
          &\!\!\! + \bar \xi \int_\mathcal{B} \nabla (\div u) \cdot \nu_t
        \end{aligned}
\end{align*}
\noindent Hence, by adding them and integrating by parts, we
  obtain
  \begin{equation}\label{eq:total-energy}
    \begin{aligned}
      \frac{\rho}{2} &\frac{d}{dt}\| u_t \|^2 + \varsigma \|\nu_t\|^2
      +
      \frac{\kappa_0}{2}\frac{d}{dt}\|\nu\|^2 \\
      & \quad \quad+ \frac{1}{2}\frac{d}{dt}\big(\mu\|\nabla u\|^2 +
      \zeta\|\nabla \nu\|^2\big)
      +\frac{1}{2}\frac{d}{dt}\big(\xi\|\div u\|^2 +
      \gamma\|\div \nu\|^2\big) \\
      & = \mu\int_{\D\mathcal{B}} u_t\cdot \frac{\D u}{\D n} +
      \zeta\int_{\D\mathcal{B}} \nu_t\cdot \frac{\D \nu}{\D n} -\kappa
      \frac{d}{dt}\int_\mathcal{B} \nabla u\cdot \nabla \nu
      -\bar{\xi} \frac{d}{dt}\int_\mathcal{B} \div u\, (\div \nu)\\
      &\quad\quad+ \xi \int_{\D\mathcal{B}} (\div u) \, u_t\cdot n +
      \gamma\int_{\D\mathcal{B}} (\div \nu) \nu_t\cdot n + \bar{\xi}
      \int_{\D\mathcal{B}} (((\div \nu)\, u_t + (\div u)\, \nu_t)\cdot n.
    \end{aligned}
  \end{equation}

  \noindent Due to the Dirichlet boundary conditions, we find $u_t(t, x)=
  0$ on $\D\mathcal{B}$ as well as $u_t(t, x)= 0$ on $\D\mathcal{B}$.
  Consequently, from the equation (\ref{eq:total-energy}) we get
  \begin{equation*}
    \begin{aligned}
      \frac{\rho}{2} &\frac{d}{dt}\| u_t \|^2 + \varsigma \|\nu_t\|^2
      +
      \frac{\kappa_0}{2}\frac{d}{dt}\|\nu\|^2\\
      & \quad \quad+ \frac{1}{2}\frac{d}{dt}\big(\mu\|\nabla u\|^2 +
      \zeta\|\nabla \nu\|^2\big)
      +\frac{1}{2}\frac{d}{dt}\big(\xi\|\div u\|^2 +
      \gamma\|\div \nu\|^2\big) \\
      & = -\kappa \frac{d}{dt}\int_\mathcal{B} \nabla u\cdot \nabla \nu
      -\bar{\xi} \frac{d}{dt}\int_\mathcal{B} \div u\, (\div \nu),
    \end{aligned}
  \end{equation*}
  and, by integrating in $(0, t)$, $t\leq T$ and exploiting H\"{o}lder's inequality, we compute
  \begin{equation*}
    \begin{aligned}
      \rho&\| u_t \|^2 + 2\varsigma \int_0^t\|\nu_t\|^2 +
      \kappa_0 \|\nu\|^2\\
      & \quad \quad+ \big(\mu\|\nabla u\|^2 + \zeta\|\nabla
      \nu\|^2\big) +\frac{}{}\big(\xi\|\div u\|^2 +
      \gamma\|\div \nu\|^2\big) \\
      & \quad \,\,\, \leq 2\kappa \|\nabla u \|\, \|\nabla \nu\|
      +2\bar{\xi} \|\div u\|\, \|\div \nu\| + \bar c,
    \end{aligned}
  \end{equation*}
  where $\bar c = \bar c(\|u_0\|_{1,2}, \|\dot u_0\|, \|\nu_0\|_{1,2},
  \rho, \kappa_0, \mu, \zeta, \xi, \bar{\xi},\gamma)$. Then, by using
  Young's inequality and rearranging the terms in the expression above, we obtain
  \begin{equation}\label{eq:stima-spaziale}
    \rho\| u_t \|^2 + 2\varsigma \int_0^t\|\nu_t\|^2 +
    \kappa_0 \|\nu\|^2
    + \frac{1}{2}\big(\mu\|\nabla u\|^2 +
    \zeta\|\nabla \nu\|^2\big) \leq \bar{c},
  \end{equation}
  which implies the inclusions $ u_t, \nabla u \in L^\infty(0, T; L^2(\mathcal{B}))$,
  $ \nu_t\in L^2(0, T; L^2(\mathcal{B}))$ and $\nu, \nabla \nu\in
  L^\infty(0, T; L^2(\mathcal{B}))$.

  \noindent By the first equation in \eqref{eq:quasicrystal}, we also have 
 
\begin{equation*} 
      \begin{aligned}
        \rho |\langle u_t (\tau) - u_t(s), \phi \rangle|\leq
        \mu \int_s^\tau & |\langle \Delta u, \phi \rangle
        + \xi \int_s^\tau |\langle\nabla \div u, \phi \rangle \\
        &+ \kappa\int_s^\tau |\langle\Delta \nu, \phi\rangle| + \bar
        \xi \int_s^\tau|\langle \nabla \div \nu, \phi \rangle |
      \end{aligned}
    \end{equation*}
 for all $\phi\in \mathcal{H}^1$ and $0\leq s\leq \tau \leq T$. By the
  boundedness of $\nabla u$ and $\nabla \nu$, which belong to $ L^\infty (0, T;
  L^2(\mathcal{B}))$, we realize that $u_t(\tau) - u_t(s)$ is bounded in
  $L^2(0, T; \mathcal{H}^{-1})$.

  \noindent Recalling that $(u, \nu)$ is actually the sequence $(u^m, \nu^m)$
  (and that $(u_t, \nu_t)$ indicates $(u^m_t, \nu^m_t)$), by using
  classical compactness arguments, we can extract a sub-sequence
  (still labeled by $(u^m, \nu^m)$) such that
  \begin{align*}
    & u^m\to \hat u \textrm{ in } \left\{\begin{array}{l} L^2(0, T;
        L^2(\mathcal{B}))
        \text{--}\textrm{strong}, \\
        L^{\infty}(0, T; W^{1,2}(\mathcal{B})) \text{--}\textrm{weak}^{\star},\\
        L^2(0, T; W^{1,2}(\mathcal{B})) \text{--}\textrm{weak},
      \end{array} \right.
    \\
    & u^m_t\to \hat{u}_t\textrm{ in }
    \left\{\begin{array}{l} L^\infty(0, T; L^2(\mathcal{B}))
        \text{--}\textrm{weak}^\star, \\
        L^{2}(0, T; L^2(\mathcal{B})) \text{--}\textrm{weak},
      \end{array} \right. 
    \\
    & u^m_{tt}\to \hat{u}_{tt}\textrm{ in } L^{2}(0, T; \mathcal{H}^{-1})
    \text{--}\textrm{weak}, 
    \intertext{and}
    &\nu^m\to \hat  \nu\textrm{ in }
    \left\{\begin{array}{l} 
        L^{\infty}(0, T; L^2(\mathcal{B})) \text{--}\textrm{weak}^{\star},\\
        L^2(0, T; W^{1,2}(\mathcal{B})) \text{--}\textrm{weak},\\
      \end{array} \right.
    \\
    & \nu^m_t \to \hat \nu_t \textrm{ in } L^2(0, T; L^2(\mathcal{B}))
    \text{--}\textrm{weak}.
  \end{align*}
  \noindent By exploiting these convergences, we can easily pass to the limit for
  the sequence $(u^m, \nu^m)$ in the weak formulation
  \eqref{eq:weak-form-1}--\eqref{eq:weak-form-2}, proving that $(\hat
  u, \hat \nu)$ is a regular weak solution to the problem
  \eqref{eq:quasicrystal}.  The continuity property of such a solution
  follows from the standard embedding of $W^{1,2} (0, T ;
  L^2(\mathcal{B}))$ in $C^\beta ([0,T]; L^2(\mathcal{B}))$ of $\beta$--H\"older
  continuous functions on $[0,T]$ with values in $L^2(\mathcal{B})$, for
  every $\beta\in (0, 1)$ (see, e.g., \cite{Adams:1975,
    Triebel:1978}). Again, the circumstance that
  $\hat u$ and $\hat u_t$ are weakly continuous with values in $\mathcal{H}^1$,
  and $L^2(\mathcal{B})$ respectively, is a direct  consequence 
of the obtained regularity, i.e. $\hat u \in  L^\infty(0,T; \mathcal{H}^1)$, $u_t \in L^2(0,T;
L^2(\mathcal{B}))$, $u_{tt}\in L^2(0,T; \mathcal{H}^{-1})$, and
the Sobolev embedding theorem.

  \noindent Uniqueness emerges from direct computations: Let $(u_1,
  \nu_1)$ and $(u_2,\nu_2)$ be two solutions of
  \eqref{eq:quasicrystal}. We take differences $U:= u_1- u_2$ and
  $V:= \nu_1-\nu_2$ and consider the related equations.  We use
  $U_t$ and $V_t$ as test functions for the equations satisfied by $U$
  and $V$, respectively. Thus, by taking the $L^2$-inner products and 
integrating in time on $(0, T)$, as in the procedure above, we
  obtain
  \begin{equation*}
    \begin{aligned}
      \rho&\| U_t \|^2 + 2\varsigma \int_0^t\|V_t\|^2 +
      \kappa_0 \|V\|^2\\
      & \quad \quad+ \big(\mu\|\nabla U\|^2 + \zeta\|\nabla V\|^2\big)
      +\frac{}{}\big(\xi\|\div U\|^2 +
      \gamma\|\div V\|^2\big) \\
      & \quad \,\,\, \leq 2\kappa \|\nabla U \|\, \|\nabla V\|
      +2\bar{\xi} \|\div U\|\, \|\div V\|,\,\, t\in (0, T),
    \end{aligned}
  \end{equation*}
  from which the conclusion follows by applying Young's inequality on
  the right-hand side terms and reabsorbing the emerging integrals.
\end{proof}

\begin{remark} \label{more-reg}
Let $(u, \nu)$ be a weak solution to \eqref{eq:quasicrystal} constructed in
Theorem~\ref{thm:linear-case}. By using standard arguments it is 
possible to show that, actually, it is such that 
$u\in C(0, T; \mathcal{H}^1)$ and $u_t\in C(0, T; L^2(\mathcal{B}))$. 
Indeed, this may be proved by taking a regularization (convolution in
time) $(u^\epsilon, \nu^\delta)$, of $(u, v)$, given by
\begin{equation*} 
u^\epsilon = \eta_\epsilon \ast  u   \in  C_0^\infty(0, T ;
\mathcal{H}^1) \textrm{ and } \nu^\delta = \eta_\delta \ast  \nu   \in  C_0^\infty(0, T ;
\mathcal{H}^1),
\end{equation*}
where the smooth function $\eta_\epsilon$ is even,
positive, supported in $(-\epsilon, \epsilon)$ and so is $\eta_\delta$ in $(-\delta,
\delta)$, with
$\int_{-\epsilon}^\epsilon  \eta_\epsilon (s) ds = 1$ ( similarly
$\int_{-\delta}^\delta  \eta_\delta (s) ds = 1$),
and using subsequently the properties of the considered system of equations along with
the convergence  $u^\epsilon\to u$ ($\nu^\delta\to \nu$)  in $L^2(0, T; \mathcal{H}^1)$ as
$\epsilon\to 0$  (as $\delta\to 0$, respectively).\\
\noindent Alternatively, as we do in analyzing the non-linear system \eqref{eq:quasicrystal-0}
below, we can use a parabolic regularization of the equations in
\eqref{eq:quasicrystal} by consider viscous components of the standard and phason stresses determining the terms $-\epsilon\Delta u_t$ and 
$-\delta\Delta \nu_t$, which appear respectively to the right-hand side of the first and second
equation in \eqref{eq:quasicrystal}. In this case, to prove the
strong continuity of the weak solution $(u, \nu)$, we can exploit the
convergence of the regularized solution $(u^\epsilon, \nu^\nu)$ to
$(u,\nu)$, as $(\epsilon, \delta)\to (0,0)$, together with the intrinsic properties of the equations
\eqref{eq:quasicrystal} (see the next section for additional details).
\end{remark}

\section{Existence results: dynamics with phason diffusion 
and non-linear gyroscopic phason inertia}
In presence of gyroscopic-type phason inertia, the system (\ref{eq:quasicrystal}) becomes
\begin{equation} \label{eq:quasicrystal-0} 
    \begin{aligned}
      &\mu \Delta u + \xi \nabla \div u + \kappa \Delta \nu + \bar \xi
      \nabla \div \nu = \rho u_{tt}
      &\textrm{in } \mathcal{B}_T,\\[0.1 cm]
      &\zeta \Delta \nu + \gamma \nabla \div \nu + \kappa \Delta u
      +\bar \xi \nabla \div u -\kappa_0 \nu= \varsigma \nu_t + \ell
      (\curl u_t)\X\nu_t
      & \textrm{in } \mathcal{B}_T,\\[0.1 cm]
      & u(t, x) = \bar u(x),\,\, \nu(t, x) = \bar \nu (x), & \textrm{on } \D\mathcal{B}_T, \\[0.1 cm]
      &u|_{t=0}=u_0, \,\, u_t|_{t=0}=\dot u_0,\,\, \nu|_{t=0}=\nu_0, &
      \textrm{on } \mathcal{B},
    \end{aligned}
\end{equation}
$\ell$ is a positive constant and $\ell ((\curl u_t)\X\nu_t)_{i}:=\ell
\mathrm{e}_{ijr} \mathrm{e}_{rhk}u_{tk|h}\nu_{tj}$, leaving understood the sum over repeated indexes, as usual. $\mathrm{e}_{rhk}$ is the $rhk$-th component of the Ricci alternating symbol $\mathrm{e}$, recalled in Section 2.

\begin{defin}[Weak solution] \label{def:weak-sol} We say that a pair
  $(u, \nu)$ is a ``weak solution'' to the system
  \eqref{eq:quasicrystal-1} if, for a given $T>0$, the following conditions hold true: \vspace{-0.3 cm}
  \begin{align}
    \intertext{Regularity:} & u \in L^\infty(0, T ; \mathcal{H}^1)
    \cap C([0, T ]; \mathcal{H}^1),\,\, \nu \in L^2(0, T ;
    \mathcal{H}^1) \cap C([0, T ];
                                  \mathcal{H}^1), \label{eq:u-nu}\\[2 mm]
    & \begin{aligned} \label{eq:Dt-u-nu}  & u_t \in
      C([0, T] ; L^2(\mathcal{B})) \cap L^2(0, T ;
      W^{1,2}(\mathcal{B})), \,\,
      u_{tt} \in L^2(0, T ; \mathcal{H}^{-1}),\\
      & \nu_t \in L^2(0, T ; W^{1,2}(\mathcal{B})).
    \end{aligned}\\
    \intertext{Weak formulation: For all $(w, h)\in
   \textcolor{black}{ C_0^\infty([0,T[\X \mathcal{B})
\X C_0^\infty([0,T[\X \mathcal{B})}$,}
    & \begin{aligned} \label{eq:weak-form-1a} \rho &\int_0^T
      \int_{\mathcal{B}} u_{tt}\, \cdot w + \mu \int_0^T \int_{\mathcal{B}}
      \nabla u \cdot \nabla w + \kappa \int_0^T \int_{\mathcal{B}} \nabla
      \nu \cdot \nabla w
      \\
      &= \int_0^T \int_{\D\mathcal{B}} w\cdot\big(\mu\frac{\D u}{\D n} +
      \kappa\frac{\D \nu}{\D n}\big) + \xi \int_0^T \int_{\mathcal{B}}
      \nabla (\div u) \cdot w + \bar \xi \int_0^T \int_{\mathcal{B}} \nabla
      (\div \nu) \cdot w,
    \end{aligned}\!\!\!\!\!\\[2 mm]
    & \begin{aligned} \label{eq:weak-form-2a} \int_0^T & \int_{\mathcal{B}}
      (\varsigma \nu_t + \kappa_0\nu)\cdot h + \ell \int_0^T
      \int_{\mathcal{B}} (\curl u_t)\X\nu_t\cdot h + \int_0^T
      \int_{\mathcal{B}} (\zeta \nabla \nu + \kappa \nabla u )  \cdot \nabla h \\
      = & \int_0^T \int_{\D\mathcal{B}} h \cdot \big(\kappa\frac{\D u}{\D
        n} + \zeta\frac{\D \nu}{\D n}\big) + \gamma \int_0^T
      \int_{\mathcal{B}} \nabla (\div \nu) \cdot h + \bar \xi \int_0^T
      \int_{\mathcal{B}} \nabla (\div u) \cdot h. \!\!\!\!\!
    \end{aligned}\!\!\!\!\!
  \end{align}
\end{defin}

\noindent To determine existence of a weak solution to \eqref{eq:quasicrystal-0},
we analyze first its regularized counterpart, obtained by introducing dissipative
 components of the stresses, fixing the parameters $\epsilon > 0$ and $\delta>0$:
\begin{equation} \label{eq:quasicrystal-1} 
\begin{aligned}
      &\mu \Delta u + \xi \nabla \div u + \kappa \Delta \nu + \bar \xi
      \nabla \div \nu = -\epsilon\Delta u_t + \rho u_{tt}
      &\ \textrm{in } \mathcal{B}_T,\\[0.1 cm]
      &\begin{aligned}
        \zeta \Delta \nu + \gamma \nabla \div \nu + \kappa \Delta u
        +\bar \xi \nabla \div u -\kappa_0 \nu= - &\delta\Delta
        \nu_t
        + \varsigma \nu_t\\
        &+ \ell (\curl u_t)\X\nu_t
      \end{aligned}
      & \textrm{in } \mathcal{B}_T,\\[0.1 cm]
      &u(t, x) = \bar u(x),\,\, \nu(t, x) = \bar \nu (x),
      & \textrm{on } \D\mathcal{B}_T, \\[0.1 cm]
      &u|_{t=0}=u_0, \,\, u_t|_{t=0}=\dot u_0,\,\, \nu|_{t=0}=\nu_0, &
      \textrm{on } \mathcal{B}.
    \end{aligned}
\end{equation}

\begin{theorem}
  Consider problem \eqref{eq:quasicrystal-0}.  Assume $\mu >  -\lambda$,
  $\kappa >0$, $\bar \xi >0$, 
  $ \mu, \zeta > 2\kappa$, and 
  $\xi, \gamma > 2\bar{\xi}$.  Assume also
  $u_0, \nu_0\in W^{1,2}(\mathcal{B})$, $\dot u_0\in W^{1,2}(\mathcal{B})$,
    such that $\ell\|\dot u_0\|_{1,2} < \varsigma/2$
  and that $\bar u, \bar \nu \in L^2(\D\mathcal{B})$.  Then, the system
  \eqref{eq:quasicrystal-0} admits a weak solution.
\end{theorem}
\begin{proof}
  First we consider the regularized model~\eqref{eq:quasicrystal-1}.
  In order to prove the existence of pertinent weak solutions $( u^\epsilon,
  \nu^\delta)$, we follow the same path leading to Theorem~\ref{thm:linear-case}. We
  apply 
  the Galerkin method by using the approximating functions
  $(u^{\epsilon, m}, \nu^{\delta, m})$; thus we proceed by testing
  the equations in \eqref{eq:quasicrystal-1} by
  $u^{\epsilon, m}_t$ and $\nu^{\delta, m}_t$, respectively.

  \noindent Due to the identity
  \begin{equation*}
    \int_\mathcal{B} (\curl u_t^{\epsilon, m})\X \nu^{\delta, m}_t \cdot
    \nu^{\delta, m}_t =0,
  \end{equation*}
  a priori estimates for the equations \eqref{eq:quasicrystal-1} are nearly the same made
  for the system \eqref{eq:quasicrystal}.

 \noindent Thus, in the case of the system \eqref{eq:quasicrystal-1},
  the following estimate holds true, provided that $\|(\nabla
  u_t^{\epsilon, m} )(0)\|$ and $\|(\nabla \nu_t^{\delta, m} )(0)\|$
  are bounded:
  \begin{equation}\label{eq:stima-spaziale-1}
    \begin{aligned}
      \rho\| u_t^{\epsilon, m} \|^2 + 2\varsigma
      \int_0^t\|\nu_t^{\delta, m}\|^2 + \kappa_0 \|\nu^{\delta, m}\|^2
      & + \int_0^t(\epsilon \|\nabla u^{\epsilon, m}_t\|^2 +
      \delta
      \|\nabla \nu^{\delta, m}_t\|^2) \\
      & + \frac{1}{2}\big(\mu\|\nabla u^{\delta, m}\|^2 +
      \zeta\|\nabla \nu^{\delta, m}\|^2\big) \leq \bar{c}.
    \end{aligned}
  \end{equation}
  Consequently, we get an improved regularity: $\nabla u_t\in L^2(0,
  T; L^2(\mathcal{B}))$ and $\nabla \nu_t\in L^2(0, T;
  L^2(\mathcal{B}))$. Here, we have $(\nabla
      u_t^{\epsilon, m} )(0)= P_m \nabla \dot u_{0}$ and the
  constant in the inequality is $\bar c = \bar c(\|u_0\|_{1,2}, \|\dot u_0\|,
  \|\nu_0\|_{1,2}, \|\nabla \dot u_{0}\|, \| (\nabla
  \nu_t^{\delta, m} )(0)\|, \epsilon, \delta, \rho, \kappa_0, \mu,
  \zeta, \xi, \bar{\xi},\gamma)$.

  \noindent In order to guarantee the validity of the estimate
  \eqref{eq:stima-spaziale-1} we have to ensure uniform a priori
  estimates for the initial datum $(\nabla \nu_t^{\delta, m} )(0)$.

  \noindent To bound $\|\nabla (\nu_t^{\delta, m} )(0)\|$, let
$\varphi \in \mathcal{H}^1$ with $\|\varphi\|_{1,2} \leq 1$.
  From the second equation in \eqref{eq:quasicrystal-1} we get
    \begin{align*}
      \varsigma|\langle \nu_t^{\delta, m} (0),   \varphi\rangle| + & 
      \delta|\langle \Delta \nu_t^{\delta, m} (0),\varphi\rangle| \\
      & = \varsigma|\langle \dot{\nu}_0^m, \Pi_m\varphi\rangle| +
      \delta|\langle \nabla (\D_t \nu^{\delta, m})(0) ,\Pi_m \nabla \varphi\rangle|\\
      & \leq |\langle \bar \xi \div u^{ m}_0 + \kappa \nabla u^m_0 +
      \gamma \div \nu_0^m + \zeta \nabla \nu_0^m
      ,\Pi_m\nabla\varphi\rangle| \\
      &\quad + \kappa_0 |\langle \nu_0^m, \Pi_m \varphi \rangle| +
      \ell |\langle \curl \dot{u}_0^m \X \dot{\nu}_0^m, \Pi_m
        \varphi\rangle|\\
        &\leq \Big( c (\|u_0\|_{1,2} + \|\nu_0\|_{1,2})\|+
        \ell \| \curl \dot u_0^m \X \dot \nu_0^m\|_{-1,
          2}\Big) \|\varphi\|_{1,2}\\
        &\leq c (\|u_0\|_{1,2} + \|\nu_0\|_{1,2}) + \ell\|\dot
        u_0\|_{1,2}\|\dot \nu_0\|.
    \end{align*} 
Since $\ell \|\dot u_0\|_{1,2}< \varsigma/2$,  we obtain
  \begin{equation*}
    \min \left\{\frac{\varsigma}{2}, \delta\right\}\|\nu^{\delta,
      m}_{t}(0)\|_{1,2} 
\leq  c (\|u_0\|_{1,2} + \|\nu_0\|_{1,2}),
  \end{equation*}
  that is $\|\nu^{\delta, m}_{t}(0)|_{1,2}\leq c (\|u_0\|_{1,2} +
    \|\nu_0\|_{1,2})$.  As a consequence, the bound \eqref{eq:stima-spaziale-1}
 depends just on $\|u_0\|_{1,2}, \|\dot
  u_0\|, \|\nu_0\|_{1,2}, \|\nabla \dot u_{0}\|, \epsilon, \delta,
  \rho, \kappa_0, \mu, \zeta, \xi, \bar{\xi}$, and $\gamma$.
  \smallskip

\noindent Consider the first equation in \eqref{eq:quasicrystal-1}. From it we get
\begin{equation} \label{eq:stima-tempo-1}
      \begin{aligned}
        \rho |\langle u^{\epsilon,m}_t (\tau) - u^{\epsilon,m}_t(s), \phi \rangle|\leq &
         \epsilon \int_s^\tau |\langle \nabla u^{\epsilon,m}_t,
        \nabla \phi_t\rangle| \\
        +&
        \mu\int_s^\tau  |\langle \nabla u^{\epsilon, m}, \nabla
        \phi \rangle|
         + \kappa\int_s^\tau |\langle\nabla \nu^{\delta, m}, \nabla \phi\rangle| \\
        +& \xi \int_s^\tau |\langle \div u^{\epsilon, m}, \div \phi \rangle |
       + \bar
        \xi \int_s^\tau|\langle  \div \nu^{\delta, m}, \div \phi \rangle |
      \end{aligned}\hspace{-0.7 cm}
    \end{equation}
for all $\phi\in \mathcal{H}^1$ and $0\leq s\leq \tau \leq T$. By the
  boundedness of $\nabla u$ and $\nabla \nu$, which belong to $L^\infty (0, T;
  L^2(\mathcal{B}))$, we find that $u_t^{\epsilon, m}(\tau) -
  u^{\epsilon, m}_t(s)$ is bounded in
  $L^2(0, T; \mathcal{H}^{-1})$.

  \noindent By exploiting the Aubin-Lions compactness argument, we obtain the
  same kind of convergences in the proof of
  Theorem~\ref{thm:linear-case} and, in addition, we realize that
  \begin{align}
    & u^{\epsilon, m}_t\to u^\epsilon_t\textrm{ in }
    \left\{ \begin{array}{l}
        L^2(0, T; W^{1,2}(\mathcal{B})) \text{--}\textrm{weak}, \\
       L^{2}([0, T]\X \mathcal{B}) \text{--}\textrm{strong},
      \end{array}\right. \label{eq:weak-u-t}
    \intertext{and} & \nu^{\epsilon, m}_t\to \nu^\delta_t\textrm{ in
    } 
\left\{ \begin{array}{l} 
          L^2(0, T; W^{1,2}(\mathcal{B}))
          \text{--}\textrm{weak},  \label{eq:weak-nu-t}\\
L^{2}([0, T]\X \mathcal{B}) \text{--}\textrm{weak}.
        \end{array} \right.
  \end{align}
As a consequence of the obtained regularity of $u_t^\epsilon$ and the interpolation theorem (see,
    e.g., \cite{Triebel:1978}), we find $ u_t^\epsilon \in C([0, T];
    L^2(\mathcal{B}))$. Moreover, due to the inclusion $W^{1,2}(0,T;
        \mathcal{H}^1)\subset C^\beta([0, T];
        \mathcal{H}^1)$, $\beta\in (0,1)$, we
        have in particular that $u^\epsilon, \nu^\delta \in C([0, T]; \mathcal{H}^1)$.

  \noindent To pass to the limit in the weak formulation, the only relevant
  point to be proved is the following: For every
  $\phi \in C_0^\infty([0,T[\X \mathcal{B})$, and $0\leq t \leq T$,
  the limit
  \begin{equation} \label{eq:limit} \left|\int_0^t\int_\mathcal{B}
      \big[(\curl u^{\epsilon, m}_t)\X \nu^{\delta, m}_t -
      (\curl u^{\epsilon}_t)\X  \nu^{\delta}_t \big]\cdot \phi\right| \to 0,
    \textrm{ as } m\to +\infty
  \end{equation}
  exists.
 In fact, we get
  \begin{equation}
    \begin{aligned}
      \int_0^t &\int_\mathcal{B} \big[(\curl u^{\epsilon, m}_t)\X
      \nu^{\epsilon, m}_t - (\curl  u^{\epsilon}_t)\X \nu^{\delta}_t
      \big]\cdot \phi \\
      =& \int_0^t\int_\mathcal{B} \nu_t^{\delta, m}\X \phi \cdot \curl
      (u^{\epsilon, m}_t - u^\epsilon_t) + \int_\mathcal{B} (
      \nu^{\delta, m}_t -\nu^{\delta}_t)\X \phi \cdot (\curl 
      u^{\epsilon}_t)\\
       =& \int_0^t\int_\mathcal{B} \curl (\nu_t^{\delta, m}\X \phi ) \cdot
      (u^{\epsilon, m}_t -  u^{\epsilon}_t)
      + \int_0^t\int_\mathcal{B} \phi \X  (\curl u^\epsilon_t) \cdot 
      ( \nu^{\delta, m}_t -\nu^{\delta}_t)\\
      =: & I_1^m + I_2^m.
    \end{aligned}
  \end{equation}

  \noindent Let us consider $I_1^m$. We have 
    \begin{equation*}
      \begin{aligned}
      \int_0^t\|\curl (\nu_t^{\delta, m}\X \phi ) \|^2 \! &=\! \int_0^t\!
       \| (\nu_t^{\delta, m} \cdot \nabla) \phi -  (\phi \cdot
      \nabla ) \nu_t^{\delta, m} + (\div \nu_t^{\delta, m} )\phi -(\div
      \phi) \nu_t^{\delta, m} \|^2\\
      &\leq c \|\nu_t^{\delta, m}\|^2_{L^2 (0, T; \mathcal{H}^1)}
      \|\phi\|^2_{L^\infty(0, T; W^{1, \infty}(\mathcal{B}))}.
    \end{aligned}
  \end{equation*}
  Since $\nu^{\delta, m}_t$ and $\nu^{\delta}_t$ 
  are uniformly bounded (with respect to $m$ and $\delta$) in $L^2([0,
  T]\X \mathcal{B})$, in view of the inequality above, it follows that 
   $\curl (\nu_t^{\delta, m}\X \phi )$ is uniformly bounded in
  $L^2([0, T]\X\mathcal{B})$.

\noindent Whence, the inequalities
  \begin{equation*}
    \begin{aligned}
      \left| \int_0^t\int_\mathcal{B} \curl (\nu_t^{\delta, m}\X \phi ) 
        \cdot  (u^{\epsilon, m}_t - u_t)\right| \leq &
      \int_0^t \|\curl (\nu_t^{\delta, m}\X \phi )\| \|u^{\epsilon, m}_t
      - u^\epsilon_t \|\\
      \leq & c \|u^{\epsilon, m}_t
      - u^\epsilon_t \|_{L^2(0, T; L^2)}
    \end{aligned}
  \end{equation*}
  hold and, by the strong convergence of  $u^{\epsilon, m}_t$ to
  $u^{\epsilon}_t$, we compute  $I^m_1\to 0$ as $m\to
  +\infty$.

  \noindent As regards the integral $I_2^{m}$, for every $\phi\in
    L^{\infty}(0, T; \mathcal{H}^1)$ we find
     \begin{equation*}
      \begin{aligned}
        \int_0^t\|\phi\X (\curl u^\epsilon_t)\|^2\leq c
        \| u^\epsilon_t\|^2_{L^2 (0, T; \mathcal{H}^1)}
        \|\phi\|^2_{L^\infty(0, T; L^\infty(\mathcal{B}))},
      \end{aligned}
    \end{equation*}
  and hence $\phi\X (\curl u^\epsilon_t)\in L^2([0, T]\X
  \mathcal{B})$. Recalling that
    \begin{equation*}
       I_2^m = \int_0^t\int_\mathcal{B} \phi \X (\curl u^\epsilon_t) \cdot
       ( \nu^{\delta, m}_t
       -\nu^\delta_t)
     \end{equation*}
  and $\nu^{\delta, m}_t$ converges to $ \nu^{\delta}_t$ weakly in
  $L^2([0, T]\X\mathcal{B})$, it follows that
    $I_2^m\to 0$ as $m\to + \infty$. Hence,
    we can pass to the limit in \eqref{eq:limit}, obtaining the
    conclusion.  \smallskip

  \noindent By the same arguments used above (essentially, by exploiting again the
  inequalities \eqref{eq:stima-spaziale-1} and
  \eqref{eq:stima-tempo-1} for the weak solution $(
  u^\epsilon, \nu^\delta)$), we can deduce that $(
  u^\epsilon, \nu^\delta)$ is uniformly bounded in $W^{1,2}(0,
  T; \mathcal{H}^{1})$ and that $u^\epsilon_{tt}$ is bounded
  in $L^2(0, T; \mathcal{H}^{-1})$.  Hence, we can pass to the limit
  as $(\epsilon, \delta)\to (0, 0)$ and we have
  \begin{align*}
    & u^\epsilon_t\to u_t\textrm{ in }
    \left\{ \begin{array}{l} L^2(0, T; W^{1,2}) \text{--}\textrm{weak}, \\
        L^{2}([0, T]\X \mathcal{B})) \text{--}\textrm{strong},
      \end{array}\right.
    \intertext{and} & \nu^\delta_t\to \nu_t\textrm{ in } 
    \left\{ \begin{array}{l}  L^2(0,
    T; W^{1,2}) \text{--}\textrm{weak},\\
              L^{2}([0, T]\X \mathcal{B}))
                            \text{--}\textrm{weak}.
            \end{array}\right.
  \end{align*}
These convergences types are enough to pass to the limit  as
$(\epsilon, \delta)\to (0, 0)$ in the weak formulation for $(
  u^\epsilon, \nu^\delta)$, and hence
  the pair $( u,  \nu)$ is a weak solution of
  \eqref{eq:quasicrystal-1} for every $T\geq 0$.
\end{proof}

\begin{remark}
The lack of uniqueness for the weak solutions to
    \eqref{eq:quasicrystal-1} is mainly related to the need of
    having 
         $\nu_t$ uniformly bounded in the $L^\infty(\mathcal{B})$-norm  in $(0, T)$. Such bound seems to be essential in estimating the difference of two possible
       solutions to \eqref{eq:quasicrystal-1}. Such requirement is a bit
       more than what is implied by the boundedness of the bulk actions
       required for proving the non-standard action-reaction principle
       satisfied by the phason tractions and the existence of the phason
       stress. The uniform boundedness of $\nu_t$ could be in principle
       reached with initial data more regular than those we have presumed
       here.

\end{remark}

\textbf{Acknowledgement}. This work has been developed within the programs
of the research group in Theoretical Mechanics of the \textquotedblleft
Centro di Ricerca Matematica Ennio De Giorgi\textquotedblright\ of the
Scuola Normale Superiore in Pisa. The support of GNAMPA-INDAM and GNFM-INDAM is
acknowledged.

\ \

\textbf{Conflict of interest:} The authors declare that they have no conflict of interest.



\begin{thebibliography}{1}
\bibitem{Adams:1975} Adams  R.\ A.\ (1975), \emph{Sobolev Spaces, Academic
    Press.}, New York.
    
\bibitem{Amazit} Amazit Y., Fischer M., Perrin B., Zarembowitch A. (1995), Ultrasonic investigations of large single AlPdMn icosahedral quasicrystals, in \emph{Proc. 5th Int. Conf. on Quasicrystals} (C. Janot and R. Mossieri Edts) (Herausgeber), pp. 584–587. Singapore: World Scientific.

\bibitem{C-N} Coleman B. D., Noll W. (1963), The thermodynamics of elastic
materials with heat conduction and viscosity, \textit{Arch. Rational Mech.
Anal.}, \textbf{13}, 245-261.
    
\bibitem{CM11} Colli S., Mariano P. M. (2011), The standard description of
  quasicrystal linear elasticity may produce non-physical results, \emph{Phys.
  Lett. A}, \textbf{375}, 3335-3339.
 
 \bibitem{DP} De P., Pelcovits R. A. (1987), Linear elasticity theory of
   pentagonal quasicrystals, \emph{Phys. Rev. B}, \textbf{35}, 8609-8620.

\bibitem{Fan} Fan T. (2011), \emph{Mathematical theory of elasticity of
  quasicrystals and its applications}, Science Press Beijing, Beijing, and
  Springer Verlag, Heidelberg.
  
  \bibitem{FMS} Focardi M., Mariano P. M., Spadaro E. (2015), Multi-value microstructural descriptors for complex materials: analysis of ground states, \emph{Arch. Rational Mech. Anal.}, \textbf{217}, 899-933.
  
  \bibitem{HWD} Hu C., Wang R., Ding D.-H. (2000), Symmetry groups, physical property tensors, elasticity and dislocations in quasicrystals, \emph{Rep. Prog. Phys.}, \textbf{63}, 1–39.

\bibitem{ICU} International Union of Crystallography. Report of the
  Executive Committee for 1991, \emph{Acta Cryst.} (1992), \textbf{A48},
  922-946.
  
  \bibitem{Jeong} Jeong H.-C., Steinhardt P. J. (1993), Finite-temperature elasticity phase transition in decagonal quasicrystals, \emph{Phys. Rev. B}, \textbf{48}, 9394–9403.
  
  \bibitem{Letoublon} Letoublon A., de Boissieu M., Boudard M., Mancini L., Gastaldi J., Hennion B., Caudron R.,
Bellissent R. (2001), Phason elastic constants of the icosahedral Al-Pd-Mn phase derived from
diffuse scattering measurements, \emph{Phil. Mag. Lett.}, \textbf{81}, 273–283.
  
  \bibitem{Li-Yun14} Li L. H., Yun G. H. (2014), Thermal stress analysis for
  octagonal quasicrystals, \emph{J. Thermal Stresses}, \textbf{37}, 429-439.
  
  \bibitem{Li} Li X. Y. (2014), Elastic field in an infinite medium of
  one-dimensional hexagonal quasicrystal with a planar crack, \emph{Int. J.
  Solids Structures}, \textbf{51}, 1442-1455.
  
  \bibitem{Lubensky} Lubensky T. C., Ramaswamy S., Toner J. (1985), Hydrodynamics of icosahedral
quasicrystals, \emph{Phys. Rev. B}, \textbf{32}, 7444–7452.
  
  \bibitem{LWWC14} Li X. Y., Wu F., Wu Y. F., Chen W. Q. (2014), Indentation
  on two-dimensional hexagonal quasicrystals, \emph{Mech. Mat.}, \textbf{76},
  121-136.
  
  \bibitem{Lian-He13} Lian-He L. (2013), Generalized 2D problem of icosahedral
  quasicrystals containing an elliptic hole, \emph{Ch. Phys. B}, \textbf{22},
  art. n. 116101.
      
\bibitem{Lions:1969} Lions  J-L.\ (1969) \emph{Quelques M\'ethodes de
    R\'esolution des Probl\`emes aux Limites Non Lin\'eaires.}, Dunod,
  Gauthier-Villars: Paris.
  
\bibitem{M06} Mariano P. M. (2006), Mechanics of quasi-periodi alloys, \emph{%
  J. Nonlinear Sci.}, \textbf{16}, 45-77.

\bibitem{M14} Mariano P. M. (2014), Mechanics of material mutations, \emph{%
  Adv. Appl. Mech.}.
  
  \bibitem{MM} Mariano P. M., Modica G. (2009), Ground states in complex
  bodies, \emph{ESAIM -- Control, Optimization and Calculus of Variations}, 
  \textbf{15}, 377-402.
  
  \bibitem{MP13} Mariano P. M., Planas J. (2013), Self-actions in
  quasicrystals, \emph{Physica D}, \textbf{249}, 46-57.
  
  \bibitem{MS15} Mariano P. M., Salvatori L. (2015), Spatial decay of the phason field in quasicrystal linear elasticity, \emph{Modelling Simul. Mater. Sci. Eng.}, \textbf{23}, 045004 (19pp).
  
  \bibitem{RM11b} Radi E., Mariano P. M. (2011), Dynamic steady state crack
  propagation in quasi-crystals, \emph{Mathematical Methods in the Applied
  Sciences}, \textbf{34}, 1-23.
  
  \bibitem{Richer} Ricker M., Bachteler J., Trebin H.-R. (2001), Elastic theory of icosahedral quasicrystals—
application to straight dislocations, \emph{Eur. Phys. J. B}, \textbf{23}, 351–363.
  
  \bibitem{RL} Rochal S. B., Lorman V. L. (2002), Minimal model of the
  phonon-phason dynamics in icosahedral quasicrystals and its application to
  the problem of internal friction in the \emph{i}-AlPbMn alloy, \emph{%
  Physical Review} B, \textbf{66}, 144204 (1-9).
  
  \bibitem{Seg1} Segev R. (2004), Fluxes and flux-conjugated stresses, in 
  \emph{Advances in Multifield Theories of Continua with Substructure}, G.
  Capriz and P. M. Mariano Edts., Birk\"{a}user, Basel, pp. 149-165.
  
  \bibitem{She} Shechtman D., Blech I., Gratias D., Cahn J. W. (1984),
  Metallic phase with long-range orientational order and no translational
  symmetry, \emph{Phys. Rev. Letters}, \textbf{53}, 1951-1954.
  
  \bibitem{Sil} \v{S}ilhav\'{y} M. (1991), Cauchy's stress theorem and tensor
  fields with divergences in L$^{p}$, \textit{Arch. Rational Mech. Anal.,} 
  \textbf{116}, 223--255.
  
  \bibitem{SvS} Schmicker, D. and van Smaalen, S. (1996), Dynamical behavior of aperiodic intergrowth
crystals, \emph{Int. J. Mod. Phys. B}, \textbf{10}, 2049–2080.
  
\bibitem{Triebel:1978} Triebel H. (1978), \emph{Interpolation Theory,
    Function Spaces, Differential Operators.}, North Holland,
  Amsterdam.
 
 \bibitem{T} Truesdell C. A. (1977), \emph{A first course in rational
  continuum mechanics}, Academic Press, New York.
  
  \bibitem{Walz} Walz C. (2003), Zur Hydrodynamik in Quasikristallen. \emph{Thesis}, University of Stuttgart.

\end{thebibliography}
\end{document}